\documentclass[12pt]{article}

\usepackage{geometry}
\usepackage{enumerate}
\usepackage{amsmath, amsfonts, amsthm, amssymb}
\usepackage{url}
\usepackage{multirow}
\usepackage{tikz}

\newtheorem{theorem}{Theorem}
\newtheorem{lemma}[theorem]{Lemma}
\newtheorem{corollary}[theorem]{Corollary}
\newtheorem{property}[theorem]{Property}
\newtheorem{definition}{Definition}

\newcounter{counter:SR}
\newcounter{counter:4-SR}
\newcounter{counter:SAT-4SR}
\newcounter{counter:BIS-3SR}
\newcounter{counter:1SR}
\newcounter{counter:SAT-3SRe}
\newcounter{counter:BIS-2SRe}
\newcounter{counter:save}

\newenvironment{remark}{\vspace{1ex}\noindent{\bf Remark.}\hspace{0.5em}}{\hfill$\Box$\vspace{1ex}}
\newenvironment{example}{\vspace{1ex}\noindent{\bf Example.}\hspace{0.5em}}{\hfill$\Box$\vspace{1ex}}

\newcommand\hide[1]{}

\def\IS{\textsc{\#IS}}

\def\calB{\mathcal{B}}
\def\calT{\mathcal{T}}
\def\sr{\textsc{SR}}
\def\sm{\textsc{SM}}

\def\SM{\textsc{\#SM}}
\def\SR{\textsc{\#SR}}
\def\kSRa{\textsc{\#$k$-attribute SR}}
\def\fourSRa{\textsc{\#$4$-attribute SR}}
\def\threeSRa{\textsc{\#$3$-attribute SR}}
\def\kSRe{\textsc{\#$k$-Euclidean SR}}
\def\threeSRe{\textsc{\#$3$-Euclidean SR}}
\def\twoSRe{\textsc{\#$2$-Euclidean SR}}
\def\oneSRa{\textsc{\#$1$-attribute SR}}
\def\BIS{\textsc{\#BIS}}
\def\SAT{\textsc{\#Sat}}
\def\poly{\mathop{\mathrm{poly}}}
\def\RHPi{\mathrm{\#RH}\Pi_1}
\def\Sat{\textsc{Sat}}

\def\APeq{\equiv_\mathrm{AP}}
\def\APred{\leq_\mathrm{AP}}

\def\From{From}

\makeatletter
\def\prob#1#2#3{\goodbreak\begin{list}{}{\labelwidth\z@ \itemindent-\leftmargin
                        \itemsep\z@  \topsep6\p@\@plus6\p@
                        \let\makelabel\descriptionlabel}
                \item[\it Name.]#1
                \item[\it Instance.]#2
                \item[\it Output.]#3
                \end{list}}
\makeatother

\title{The Complexity of Approximately Counting Stable Roommate Assignments}
\author{Prasad Chebolu\footnote{Department of Computer Science, University of Liverpool,
      Ashton Bldg, Ashton St,
      Liverpool L69 3BX, United Kingdom.}\
     \footnote{Research supported in part by EPSRC Grant EP/F020651/1.}
\and Leslie Ann Goldberg\footnotemark[1]\ \footnote{Research supported in part by EPSRC Grant EP/I011528/1.}
\and Russell Martin\footnotemark[1]\ \footnotemark[2]}\date{\today}

\begin{document}
\maketitle

\begin{abstract}
We investigate the complexity of {\em approximately counting} stable
roommate assignments in two models: (i) the $k$-attribute model, in
which the preference lists are determined by dot products of
``preference vectors'' with ``attribute vectors'' and (ii) the
$k$-Euclidean model, in which the preference lists are determined by
the closeness of the ``positions" of the people to their ``preferred
positions". Exactly counting the number of assignments is
$\#P$-complete, since Irving and Leather demonstrated
$\#P$-completeness for the special case of the stable marriage
problem~\cite{IL}.  We show that counting the number of stable
roommate assignments in the $k$-attribute model ($\kSRa$, $k \geq
4$) and the $3$-Euclidean model($\kSRe$, $k \geq 3$) is
interreducible, in an approximation-preserving sense, with counting
independent sets (of all sizes) ($\#IS$) in a graph, or counting the
number of satisfying assignments  of a Boolean formula ($\#SAT$).
This means that there can be no FPRAS for any of these problems
unless NP=RP. As a consequence, we infer that there is no FPRAS for
counting stable roommate assignments ($\#SR$) unless NP=RP.
Utilizing previous results by the authors~\cite{CGM}, we give an
approximation-preserving reduction from counting the number of
independent sets in a bipartite graph ($\#BIS$) to counting the
number of stable roommate assignments both in the $3$-attribute
model and in the $2$-Euclidean model. \#BIS is complete with respect
to approximation-preserving reductions in the logically-defined
complexity class $\RHPi$. Hence, our result shows that an FPRAS for
counting stable roommate assignments in the $3$-attribute model
would give an FPRAS for all of $\RHPi$. We also show that the
1-attribute stable roommate problem always has either one or two
stable roommate assignments, so the number of assignments can be
determined exactly in polynomial time.
\end{abstract}

\section{Introduction}

The {\em stable roommate problem} is a generalization of the classical
{\em stable marriage problem}.  An instance of the roommate problem
consists of $2n$ people, where each person has a strict preference
ordering (a total ordering) of the other $2n-1$ people.  A {\em matching}
is a pairing of the people into $n$ pairs, and a matching is said to be
{\em stable} if there does not exist a pair of two people $P_1$ and $P_2$,
each of whom prefers the other over their current partners in
the matching.  Such a pair is referred to as a {\em blocking pair}
as $P_1$ and $P_2$ would drop their current partners and pair up together.

 The stable marriage problem is  the special case in which the $2n$ people consist of $n$ men
and $n$ women, and each man ranks all of the women higher
than any other man and, similarly, each women ranks all of the men
higher than any other woman.  (This is not the usual definition
of the stable marriage problem, but is equivalent to the standard one.)
In 1962, Gale and Shapley proved that every
stable marriage instance has a stable matching, and described an
$O(n^2)$ algorithm for finding one~\cite{GS}.
The stable marriage problem, including many variants, has
seen much study as algorithms for finding
stable matchings are used for assigning residents to hospitals in
Scotland, Canada, and the USA~\cite{CRMS,NRMP,SFAS}.

More than  twenty years after Gale and Shapley's seminal paper,
  Irving provided an efficient algorithm for the
stable roommate problem~\cite{Irving}.  In contrast to
the marriage problem, an instance of the roommate problem need
not have any stable matching, as this example (which may
be found in both~\cite{Irving} and~\cite{Knuth}) demonstrates:
\begin{center}
\begin{tabular}{c|lll}
Person  & \multicolumn{3}{l}{Preference list} \\
\hline
A  & B & C & D \\
B  & C & A & D \\
C  & A & B & D \\
D  & \multicolumn{3}{l}{arbitrary} \\
\end{tabular}
\end{center}

Irving's polynomial-time algorithm
determines whether a stable roommate assignment exists for the given instance,
and constructs a stable assignment if one exists.

In what follows, we will abbreviate ``stable roommate problem''
and ``stable marriage problem'' as \sr\ and \sm, respectively.

Since the problem of determining whether a stable assignment exists
is solved  in both the stable roommate setting and the stable matching setting,
it is natural to ask whether it is feasible to determine {\em how many} stable assignments
there are for a given instance.   We denote
these counting versions of \sr\ and \sm\ as $\SR$ and $\SM$, respectively.

Irving and Leather~\cite{IL} demonstrated that
\SM\
(counting the {\em number}
of stable matchings for a given \sm\ instance) is $\#P$-complete.  This
completeness result relies on the connection between stable
marriages and {\em downsets} in a related partial order and on the fact that
counting downsets in a partial order is  $\#P$-complete~\cite{PB}. As \SM\ is a restricted version of \SR, we can
obviously  conclude that \SR\ is  $\#P$-complete.

 Since exactly counting stable matchings is difficult (under
standard complexity-theoretic assumptions),  it would be good to have algorithms for {\em approximately} counting.  In
particular, we would like to find a {\em fully-polynomial randomized
approximation scheme} (an {\em FPRAS}) for this task, i.e.\ an algorithm
that provides an arbitrarily close approximation in time polynomial
in the input size and the desired error.

Randomized approximation schemes have proven successful (sometimes under
certain restrictions or conditions) for problems
such as counting the number of (perfect) matchings in bipartite graphs,
the number of proper $k$-colorings of graphs, and the number of linear
extensions of a partial order.  Many of these approximation schemes
rely on the Markov Chain Monte Carlo (MCMC) method.  This technique also
exploits a
relationship between counting and sampling described by Jerrum, Valiant,
and Vazirani~\cite{JVV}, namely, for {\em self-reducible} combinatorial
structures, the existence of an FPRAS is computationally equivalent to
a polynomial-time algorithm for approximate sampling from the set of
structures.

Bhatnagar, Greenberg, and Randall~\cite{Randall} considered the problem
of sampling a random stable matching for the  stable marriage problem
 using the MCMC method.
They examined a natural Markov chain that uses ``male-improving''
and ``female-improving'' {\em rotations} (see Section~\ref{sect:poset}
for similar definitions in the context of the roommate problem)
to define a random walk on the state space of stable matchings for a given
instance.  In the most general setting, matching instances can be exhibited for
which the {\em mixing time} of the random walk has an exponential lower
bound, meaning that it will take an exponential amount of time to
(approximately) sample a random stable matching.  This exponential mixing time is
due to the existence of a ``bad cut'' in the state space.
Bhatnagar, et al.\ considered several restricted settings for matching
instances and were still able to show instances for which such a bad
cut exists in the state space, implying an exponential mixing time
in these restricted settings.

One of the special cases   that Bhatnagar et al.\ considered was
the so-called {\em $k$-attribute model}.
In this setting, each man
and woman has two $k$-dimensional vectors associated with them, a
``preference'' vector and a ``position'' (or ``attribute'') vector.
A man $M_i$ has a preference vector denoted by $\hat{M}_i$,
and a position vector denoted by $\bar{M}_i$.
Similarly, a woman $w_h$ has a preference vector $\hat{w}_j$ and a position vector $\bar{w}_j$.
  Then, $M_i$ prefers $w_j$ over $w_k$ (i.e.\
$w_j$ appears higher on his preference list than $w_k$) if and
only if $\hat{M}_i \cdot \bar{w}_j > \hat{M}_i\cdot \bar{w}_k$, where
$\hat{M}_i \cdot \bar{w}_j$ denotes the usual $k$-dimensional dot
product of vectors.\footnote{
Recall that the dot product of two vectors ${\mathbf a}=(a_1,\ldots,a_k)$
and ${\mathbf b}=(b_1,\ldots,b_k)$ is the sum $\sum_{i=1}^k a_i b_i$
which is equal to $||{\mathbf a}|| \> ||{\mathbf b}|| \cos \theta$,
where $||{\mathbf x}||$ denotes the length of a vector~$x$ and
$\theta$ is the angle between~$\mathbf a$ and~$\mathbf b$.}
Since we assume that each man has a total order
over the women (and vice-versa),  a valid instance has the property that
$\hat{M}_i \cdot \bar{w}_j \not= \hat{M}_i \cdot \bar{w}_k$ whenever
$j \not= k$ (and analogously for the women's preference vectors/men's
position vectors).
In this paper we consider the $k$-attribute model for the roommate problem.

We also study the stable roommate problem in the {\it $k$-Euclidean
} model which we had introduced in a previous paper~\cite{CGM}. In
the $k$-Euclidean model, each person has two  associated points in $k$-dimensional
Euclidian space
  --- a ``preference'' point and a ``position" point. The
preference point of a person $X$ is denoted by $\hat{X}$, and the
position point is denoted by $\bar{X}$. Then, $X$ prefers $y$ over
$z$ (i.e.\ $y$ appears higher on his/her preference list than $z$)
if and only if $|\hat{X} - \bar{y}| < |\hat{X}- \bar{z}|$, where
$|\hat{X} - \bar{y}|$ denotes the usual $k$-dimensional Euclidean
distance. Once again, a valid instance has the property that
 $|\hat{X} - \bar{y}| \neq
|\hat{X}- \bar{z}|$ whenever $j \not= k$.

 We examined the stable marriage problem in our previous
paper~\cite{CGM}, providing complexity-theoretic evidence for the   difficulty of
approximately counting stable matchings in  both   the $k$-attribute model and   the $k$-Euclidean
model.  We constructed {\em approximation-preserving reductions}
between (i) counting the number of stable matchings in the
$k$-attribute marriage problem ($k \geq 3$) and
counting independent sets in a bipartite graph ($\#BIS$), and (ii)
counting the number of stable marriages in the $k$-Euclidean
marriage problem ($k \geq 2$) and  $\#BIS$.

Informally
speaking, if there is an approximation-preserving reduction (AP-reduction) from one problem to another, then an FPRAS for
the second problem implies the existence of an FPRAS for the first. We write $f \leq_{AP} g$ to
mean that $f$ has an AP-reduction to $g$.  Similarly, we write $f
\equiv_{AP} g$ to mean that $f \leq_{AP} g$ {\em and} $g \leq_{AP}
f$, or that $f$ and $g$ are AP-interreducible.
Approximation-preserving reductions play a role in approximate counting
analogous to the role that polynomial many-one reductions play in the theory of NP-completeness
and polynomial Turing reductions play in the theory of $\#P$-completeness.

The complexity class $\RHPi$ of counting problems was
introduced by Dyer, Goldberg, Greenhill and Jerrum~\cite{DGGJ} as a
means to classify  approximate counting problems.
The problems in $\RHPi$ are those that can be expressed in terms of
counting the number of models of a logical formula from a
certain syntactically restricted class which is also known as
``restricted Krom SNP''~\cite{Dalmau05}. The complexity class
$\RHPi$ has a completeness class (with respect to AP-reductions)
which includes  many natural
counting problems including:  $\#BIS$, counting downsets in a partial order, counting configurations in
the Widom-Rowlinson model (all~\cite{DGGJ}) and computing the partition
function of the ferromagnetic Ising model with a mixed external
field~\cite{ising}. Either all of these problems have an
FPRAS, or none do.  No FPRAS is currently known for any of them,
despite much effort having been expended on finding one.
More background and details about AP-reducibility are given in
Section~\ref{sect:AP}.

Before we continue, we define  the problems that are of interest to
us in this paper.

\prob{$\#SR$}{A stable roommate instance with $2n$
people.}{The number of stable roommate assignments.}

\prob{$\#k$-attribute SR}
{A stable roommate instance with $2n$ people,
i.e.\ preference lists are determined
using dot products between $k$-dimensional preference and
position vectors as  described above.}
{The number of stable roommate assignments.}

\prob{$\#k$-Euclidean SR.}
{A stable roommate instance with $2n$
people.  In this setting, each person has a ``preference point'' and
``position point''.  Preference lists are determined using Euclidean
distances between preference points and position points as  described
above.}{The number of stable roommate assignments.
}

We also define two other counting problems which are relevant to our results.

 \prob{$\IS$.} {A graph $G$.} {The number of independent sets (of all sizes) of $G$.}
 \prob{$\#\Sat$.}{A boolean formula in conjunctive normal form.}{The number of satisfying assignments.}

\subsection{Our results}\label{sect:results}

Zuckerman \cite{zuckerman} has shown that $\SAT$ cannot have an FPRAS unless NP=RP.
The same is true of any problem in \#P to which $\SAT$ is AP-reducible~\cite{DGGJ}. For example,
it is true   of $\IS$, which is AP-interreducible with $\SAT$ \cite{DGGJ}. We have the following results.

\setcounter{counter:4-SR}{\value{theorem}}
\begin{theorem}\label{thm:4-SR}
$\IS \APeq \kSRa$ for $k\geq 4$.
\end{theorem}

\setcounter{counter:SAT-3SRe}{\value{theorem}}
\begin{theorem}\label{thm:3SRe}
$\IS \APeq \kSRe$ for $k\geq 3$.
\end{theorem}

\setcounter{counter:SR}{\value{theorem}}
\begin{corollary}\label{thm:SR}
For any $k\geq 4$,
$\kSRa$ is complete for \#P with respect to AP-reductions.
 For any $k\geq 3$,
  $\kSRe$ is complete for \#P with respect to AP-reductions.
  Also, $\SR$ is complete for \#P with respect to AP-reductions.
  None of these problems has an FPRAS unless NP=RP.
\end{corollary}

\setcounter{counter:1SR}{\value{theorem}}
\begin{theorem}\label{thm:1SR}
For every $\oneSRa$ instance~$I$, there are either~$1$ or~$2$ stable
assignments. Thus, $\oneSRa$ can be solved exactly in polynomial time.
\end{theorem}

We also show the following results.\footnote{
The proofs of Theorems~\ref{thm:BIS-3SR} and~\ref{thm:BIS-2SRe}
(in Section~\ref{sec:BIShardness}) borrow constructions from the AP-reductions that we presented
in~\cite{CGM} from $\#BIS$ to the problem of counting stable matchings in the 3-attribute
model and the 2-Euclidean model. However, Theorems~\ref{thm:BIS-3SR} and~\ref{thm:BIS-2SRe}
do not follow directly from the results of~\cite{CGM} since, in the stable roommate problem,
all people need to rank \emph{all} other people
(rather than just ranking people of the opposite sex) and this needs to be incorporated into the geometric constructions.
}

\setcounter{counter:BIS-3SR}{\value{theorem}}
\begin{theorem}\label{thm:BIS-3SR}
$\BIS \APred \threeSRa$.
\end{theorem}

\setcounter{counter:BIS-2SRe}{\value{theorem}}
\begin{theorem}\label{thm:BIS-2SRe}
$\BIS \APred \twoSRe$.
\end{theorem}

The last two results are significant since $\BIS$ is complete for $\RHPi$ with respect to approximation-preserving
reductions.

\section{Randomized Approximation Schemes and \\ Approximation-preserving reductions}\label{sect:AP}

In this section, we give standard definitions of randomized approximation schemes and AP-reductions.
A reader who is already familiar with these concepts may safely skip this section.

A {\em randomized approximation scheme} is an algorithm for
approximately computing the value of a
function~$f:\Sigma^*\rightarrow \mathbb{R}$. The approximation
scheme has a parameter~$\varepsilon>0$ which specifies the error
tolerance. A {\em randomized approximation scheme} for~$f$ is a
randomized algorithm that takes as input an instance $ x\in
\Sigma^{* }$ (e.g., for the problem $\#SR$, the input would be an
encoding of a stable roommate instance) and a rational error
tolerance $\varepsilon >0$, and outputs a rational number $z$ (a
random variable of the ``coin tosses'' made by the algorithm) such
that, for every instance~$x$,
\begin{equation}
\label{eq:3:FPRASerrorprob}
\Pr \big[e^{-\epsilon} f(x)\leq z \leq e^\epsilon f(x)\big]\geq \frac{3}{4}\, .
\end{equation}
The randomized approximation scheme is said to be a
{\em fully polynomial randomized approximation scheme},
or {\em FPRAS},
if it runs in time bounded by a polynomial
in $ |x| $ and $ \epsilon^{-1} $.

We now define the notion of an approximation-preserving (AP) reduction.
Suppose that $f$ and $g$ are functions from
$\Sigma^{* }$ to~$\mathbb{R}$. As mentioned before,
an AP-reduction from~$f$ to~$g$ gives a way to turn an FPRAS for~$g$
into an FPRAS for~$f$. Here is the formal definition.
An {\em approximation-preserving reduction}
from $f$ to~$g$ is a randomized algorithm~$\mathcal{A}$ for
computing~$f$ using an oracle for~$g$. The algorithm~$\mathcal{A}$ takes
as input a pair $(x,\varepsilon)\in\Sigma^*\times(0,1)$, and
satisfies the following three conditions: (i)~every oracle call made
by~$\mathcal{A}$ is of the form $(w,\delta)$, where
$w\in\Sigma^*$ is an instance of~$g$, and $0<\delta<1$ is an
error bound satisfying $\delta^{-1}\leq\poly(|x|,
\varepsilon^{-1})$; (ii) the algorithm~$\mathcal{A}$ meets the
specification for being a randomized approximation scheme for~$f$
(as described above) whenever the oracle meets the specification for
being a randomized approximation scheme for~$g$; and (iii)~the
run-time of~$\mathcal{A}$ is polynomial in $|x|$ and
$\varepsilon^{-1}$.

According to the definition, approximation-preserving reductions may use randomization
and may make multiple oracle calls. Nevertheless,
the reductions that we present in this paper are deterministic.
Each reduction makes a single oracle call (with $\delta=\epsilon$) and returns the
result of that oracle call.
A word of warning about terminology:
Subsequent to~\cite{DGGJ}, the notation $\APred$ has been
used
to denote a different type of approximation-preserving reduction which applies to
optimization problems.
We will not study optimization problems in this paper, so hopefully this will
not cause confusion.

\section{Background and definitions}\label{sect:background}
We first review some of the relevant background and definitions
related to stable matchings.
The combinatorial structure present in these problems
plays a large role in what follows.  Many of the definitions are
taken from~\cite{Irving} and~\cite{gusfield-structure}.
The reader is also referred to Gusfield and Irving's book \cite{GusfieldIrving}.

It will also help to have an illustrative example, and for these purposes
we give such an example in the Appendix.

\subsection{Stable matchings and the rotation poset}\label{sect:poset}

Irving's method for finding a stable matching for a \sr\
instance (or concluding that one
doesn't exist) is a two-phase algorithm~\cite{Irving}.
During both phases of the algorithm, the preference lists are shortened
in a well-defined manner.  If we reach a stage where each person has
a single element on his/her list, then pairing these people will create
a stable matching.  Alternatively, if at any point a person's preference
list becomes empty, we conclude that the instance has no stable matching.

Phase 1
is much akin to the usual Gale-Shapley algorithm for the marriage
problem, in that people ``propose'' to one another, holding the best proposal
from the ones received so far.
For every person~$e_i$, let
$h_i$ denote the person who is (currently) first on  $e_i$'s list.
Following~\cite{gusfield-structure}, we will say that   $e_i$ is
{\em semi-engaged} to $h_i$ if and only if $e_i$ is the bottom entry of
$h_i$'s list.  Note that this is not a symmetric relation ---
$h_i$
is not necessarily semi-engaged in this case.
A person who is not semi-engaged is called {\em free}.

\smallskip

Phase 1 of Irving's \sr\ algorithm consists of the following steps:
\begin{enumerate}
\item If there is an empty list, then stop, there is no stable assignment.
\item Otherwise, if everyone is semi-engaged, go to Phase 2 (described below).
\item Otherwise, pick an arbitrary free person $e_i$ and do the following:
For each person $p$ who is ranked below $e_i$ on $h_i$'s list, remove
$p$ from $h_i$'s list, and remove $h_i$ from $p$'s list.
\end{enumerate}

So as Phase 1 proceeds, people's preference lists shrink.
At any point during this phase (or the next), we refer to the
shortened preference lists   as
  {\em short lists}, and  we  refer to the set of short lists
as a {\em table}. It is proved in~\cite{Irving} that if a short list
is empty at the end of Phase 1, then the instance has no
stable matching. An example of Phase~1 is in the appendix.

Assuming Phase 1 ends with no empty short list, we proceed to Phase 2.
To describe this phase, we need more notation and definitions.
For a person $e_i$, we are already using $h_i$ to denote the person
at the head of her short list, and we use $s_i$ to denote the person who
is second on her short list.

\begin{definition}\label{def:rotation}
Given a set of short lists, a {\em rotation} $R$ is an ordered set of
people $E=\{e_1, e_2, \ldots, e_k \}$ such that $s_i = h_{i+1}$ for
all $i \in \{1, \ldots, k-1\}$
and $s_{k}=h_1$.
We
will also say that $R$ is {\em exposed} in the short lists.
\end{definition}

Note that rotations are defined relative to a given set of short lists.
An example is in the appendix.

For a rotation $R$, we will sometimes write $R=(E, H, S)$,
where $H$ is the set of head entries
of $E$, ordered in correspondence with $E$, and, $S$ is the
set of second entries of $E$, again ordered in correspondence with $E$.

\begin{definition}\label{def:rotation-elimination}
Given a rotation $R=(E, H, S)$ for a set of short lists,
the {\em elimination} of $R$ consists of performing the following
operation: for every $s_i \in S$, remove every entry below $e_i$ in
$s_i$'s short list, i.e.\ move the bottom of $s_i$'s shortlist
up to $e_i$.  Then remove $s_i$ from $p$'s list for each person $p$
that was just removed from $s_i$'s list.
\end{definition}

Therefore, the elimination of a rotation results in a new set of short lists,
where at least two people's lists have shrunk in length.

\smallskip

Phase 2 of Irving's \sr\ algorithm
consists of the following steps:
\begin{enumerate}
\item If a short list is empty, then stop, the instance has no stable matching.
\item Otherwise, if each person has exactly one entry on his or her short list,
then pairing each person with their head entry is a stable matching.
\item Otherwise, find and eliminate some rotation.
\end{enumerate}

For an example of one round of Phase~2, see the appendix.
We note the following property of the short lists, which is easily
established from the Phase 1 and Phase 2 procedures.

\begin{property}\label{prop-1} At the end of Phase 1, and at the end of each
round of Phase 2, person $A$ has person $B$ on
his/her list if and only if person $B$ has person $A$ on
his/her list.
\end{property}

A \sr\ instance may have many stable matchings.  Each such stable matching can be found as a result of
{\em some} sequence of rotation eliminations~\cite{gusfield-structure}.

The set of rotations exhibits a rich combinatorial structure which
has been explored previously by other authors.  We review this
structure here. To do so, we need still more notation and
definitions.

\begin{definition}\label{def:dual-rotation}
Suppose that $R=(E, H, S)$ is a rotation for a \sr\ instance, i.e.\
$R$ is exposed in some set of short lists.  Define $R^d$ to be
the triple $(S, E, E^r)$, where $S$ and $E$ have the same order as
they do in $R$, and $E^r$ is the backwards cyclic rotation
of $E$ That is, if $E=\{ e_1, e_2, \ldots, e_k \}$ then
$E^r = \{ e_2, \ldots, e_k, e_1 \}$.

$R^d$ has the form of a rotation.  If $R^d$ is actually a rotation
(i.e.\ $R^d$ is exposed in the set of short lists during some possible
execution of the matching algorithm), then we call $R$ and $R^d$
a {\em dual pair} of rotations.  Any rotation without a dual is
called a {\em singleton} rotation.
\end{definition}

An ordering relation can be defined on the set of all rotations
(singletons and dual pairs).

\begin{definition}\label{def:explicitly-precedes}
A rotation~$R'$
{\em explicitly precedes} a rotation~$R$
if there is a person~$p$ who satisfies both of the following.
\begin{itemize}
\item $R$ contains a triple $(e_i,h_i,s_i)$
with $h_i\neq p$
such that $p$ is above~$s_i$ in $e_i$'s
(original) preference
list.
\item $R'$ removes~$p$ from $e_i$'s list
by moving the end of $p$'s list above~$e_i$.
\end{itemize}
\end{definition}

\begin{definition}\label{def:precedes}
$\Pi^*$ is the reflexive transitive closure of the
``explicitly precedes'' relation. We will use the term
{\em ``precedes''} to refer to this relation. $\Pi^*(R,R')$ means
that rotation $R$ precedes $R'$ in this partial order.
\end{definition}

Let $Rot(I)$ denote the set of all rotations (singletons and dual pairs)
that are exposed during some execution of the two-phase algorithm
for a given stable roommate instance $I$.
Then $\Pi^*$ defines a partial order on $Rot(I)$.  We refer to this
partial order as the {\em rotation poset}.  As usual when dealing
with partial orders, a subset $U \subseteq Rot(I)$ is
called a {\em downset} if $R' \in U$ and $\Pi^*(R,R')$ imply that
$R\in U$.

The combinatorial significance of the rotation
poset is captured in the following theorem.

\begin{theorem}\cite[Thm 5.1]{gusfield-structure}\label{thm:rotation-downsets}
There is a one-to-one correspondence between stable matchings and
the downsets in $\Pi^*$ that contain every singleton rotation and
exactly one of each dual pair.
\end{theorem}

The rotation poset $\Pi^*$ has even more structure to it.

\begin{lemma}\cite[Lemma 5.5]{gusfield-structure}\label{lem:rotation-structure}
Let $\{R_1, R_1^d\}$ and $\{R_2, R_2^d\}$ be two dual pairs of
rotations and $R$ a singleton rotation.  Then
\begin{enumerate}
\item Neither $R_1$ nor $R_1^d$ precedes $R$ in $\Pi^*$, i.e.\
only a singleton rotation can precede a singleton.
\item $\Pi^*(R_1, R_2)$ if and only if $\Pi^*(R_2^d, R_1^d)$.
\end{enumerate}
\end{lemma}

The rotation poset plays a key role in our  approximation-preserving reductions.  For our
reductions, we must define a roommate instance $I$, then identify
$Rot(I)$, find the precedence relations amongst them (i.e.\ find
$\Pi^*$), and show that it agrees with our initial starting problem.

 Given an instance $I$ of \sr\ with $2n$ people,
Gusfield describes a polynomial-time (in $n$) algorithm for finding
the set of all rotations of $I$, and for constructing a directed
graph $D$ that captures the partial order
$\Pi^*$~\cite{gusfield-structure}. (Note: The transitive reduction
of $D$ is isomorphic to the Hasse diagram of $\Pi^*$, but $D$ might
contain more edges than the covering relations defined by the
``explicitly precedes'' relation.  Still, $D$ has no more than
$O(n^2)$ edges.)

Before we demonstrate our constructions, we note one more
combinatorial construction that serves to encode the
set of stable matchings for a given instance.

\subsection{Stable matchings and independent sets}\label{sect:ind-sets}

Gusfield defines an additional way to represent the set of
stable roommate
assignments~\cite[Section 5.3.2]{gusfield-structure}.

Let $I$ denote an instance of \sr.  Define an
undirected graph $G(I)$ as follows:  Each nonsingleton
rotation of $I$ corresponds to a vertex of $G(I)$.
Two rotations $R_1$ and $R_2$ are connected by an edge in
$G(I)$ if and only if there exists a rotation $R$ (possibly
$R_1$ or $R_2$ themselves) such that $\Pi^*(R, R_1)$ and
$\Pi^*(R^d, R_2)$.  In particular, we note that $R_1$ and
$R_1^d$ are connected
by an edge for each dual pair $\{R_1, R_1^d\}$ (as a vertex
precedes itself, by definition, in the partial order $\Pi^*$). See the appendix for an example.

Gusfield also defines another partial order involving only the
non-singleton rotations.

\begin{definition}Suppose $\Sigma$ is the set of all singletons in $\Pi^*$. Then
$\Pi = \Pi^*-\Sigma$ is a partial order on the set of all
non-singleton (dual) rotations.
\end{definition}

Having defined this undirected graph and the partial order $\Pi$, we
have these results, a combination of Lemmas 5.6 and 5.10, and
Theorem 5.3 in~\cite{gusfield-structure}.

\begin{theorem}\label{thm:indsets}
Let $I$ denote an instance of \sr\ and $G(I)$ its corresponding
graph constructed as above.
\begin{enumerate}
\item Every maximal independent set in $G(I)$ contains exactly one
node from each dual pair of rotations.
\item There is a one-to-one correspondence between maximal
independent sets in $G(I)$ and stable matchings of $I$.
\item Rotations $R_1$ and $R_2$ are connected by an edge in $G(I)$
if and only if $R_1^d$ precedes $R_2$, i.e.\ $R_1$ and $R_2$ are
connected if and only if $\Pi(R_1^d, R_2)$.
\end{enumerate}
\end{theorem}

\section{A construction for showing  $\IS\APeq$\\ $\fourSRa$}\label{section:begin-construction}

Recall that $\kSRa$ denotes the problem of counting stable
assignments for $k$-attribute stable roommate instances.
Our goal of  this section is to prove Theorem~\ref{thm:4-SR}
which we restate below.

\setcounter{counter:save}{\value{theorem}}
\setcounter{theorem}{\value{counter:4-SR}}
\begin{theorem}
$\IS \APeq \kSRa$ for $k\geq 4$.
\end{theorem}
\setcounter{theorem}{\value{counter:save}}

\begin{proof}

First, since $\IS$ is complete for $\#P$ with respect to AP-reductions~\cite{DGGJ}, and $\kSRa \in \#P$,
we immediately have
 $\kSRa \APred \IS$.
Also, it is easy to see, for $k>4$, that $\fourSRa \APred \kSRa$
(the reduction uses up the extra $k-4$ dimensions by assigning some
particular value in every preference vector and attribute vector).
Thus, it remains to prove $\IS \APred \fourSRa$. This is proved in
the rest of Section~\ref{section:begin-construction}, including
Subsections~\ref{sec:task} --- \ref{sect:Stocktaking}.
\end{proof}

We wish to prove $\IS \APred \fourSRa$. To this end,
let~$\Gamma=(V,E)$ be an instance of~\IS. Let $\calT$ be the
Cartesian product $\calT = V \times \{0\}$ and let $\calB = V \times
\{1\}$. The sets $\calT$ and $\calB$ are just two distinct copies of
$V$. Let $E'$ be the matching on~$\calB\cup \calT$ defined by $E' =
\{ ((v,0),(v,1)) \mid v\in V \}$. Let $E'' = \{ ((v,0),(w,0)) \mid
(v,w) \in E \}$. The set $E''$ just mimics the edges of~$E$ amongst
the vertices of~$\calT$. Finally, let $\Gamma'=( \calB\cup \calT,E'
\cup E'')$. Note the bijection between independent sets of~$\Gamma$,
(which we wish to approximately count, using as an oracle a FPRAS
for $\fourSRa$), and maximal independent sets of~$\Gamma'$. Our goal
is to construct a $4$-attribute stable roommate instance~$I$ so that
the graph~$G(I)$ is isomorphic to~$\Gamma'$. This will complete the
proof, by  Theorem~\ref{thm:indsets}. So from now on, we will focus
exclusively on constructing~$I$ so that~$G(I)$ is isomorphic
to~$\Gamma'$. As soon as we've done that, we are finished.

We recall the construction of~$G(I)$ from~$I$: The vertices of~$G(I)$ are
the nonsingleton rotations of~$I$. From Theorem~\ref{thm:indsets}, two rotations~$R_1$ and~$R_2$ are
connected by an edge of~$G$ if and only if $\Pi(R_1^d, R_2)$ (and hence, by Lemma~\ref{lem:rotation-structure}, $\Pi(R_2^d, R_1)$).

Our method, given $\Gamma'$, will be to construct~$I$ so that its nonsingleton rotations
can be labelled bijectively with $\calB\cup \calT$
in such a way that the following conditions are satisfied (by the partial order $\Pi$ associated with $I$).
\begin{equation}\label{eq:goal0}
\textrm{For all } v\in V,
\textrm{ the rotation labelled } (v,0)
\textrm{ is dual to the one labelled } (v,1).
\end{equation}
Also, using the vertices of~$\Gamma'$ to refer to the corresponding rotations of~$I$,
\begin{align}
\label{eq:goal2} \{ (R,R') \in \calT \times \calT \mid \Pi(R^d,R') \} &= E''
,\\
\label{eq:goal3} \{ (R,R') \in \calT \times \calT \mid \Pi(R,R'^d) \} &=
\emptyset
,\\
 \label{eq:goal1}
\{ (R,R') \in \calT \times \calT \mid \Pi(R , R') \} &= \{ (R,R') \in \calT
\times \calT \mid R=R'\}.
\end{align}

Under the assumption that~(\ref{eq:goal0}) holds,
\begin{itemize}
\item Equation~(\ref{eq:goal2})
guarantees that $E(G(I)) \cap (\calT \times \calT) = E''$.
\item Equation~(\ref{eq:goal3}) guarantees that $E(G(I)) \cap (\calB \times \calB)
= \emptyset$.
\item Finally, equation~(\ref{eq:goal1})
  guarantees that $E(G(I))\cap (\calB \times \calT) = E'$.
  \end{itemize}
Thus, Equations~(\ref{eq:goal0})--(\ref{eq:goal1}), taken together, guarantee that $G(I)$ is isomorphic to~$\Gamma'$, as required.
So all that we need to do, to complete the reduction, and the proof,
is   to use the input graph~$\Gamma$ to construct an instance~$I$ so
that its nonsingleton rotations can be labelled bijectively with
$\calB\cup \calT$ in such a way that Equations
(\ref{eq:goal0})--(\ref{eq:goal1}) are satisfied. We concentrate on this for the rest of the proof.
(We never need to consider $\Gamma'$ again.)

Unfortunately, the notation
that we used to get this far (which came from earlier papers)
is not going to be very convenient when we
come to actually do the construction.
So, in order to make the following (rather complicated!) proof easier to follow, we
are now going to change notation.

First, instead of using the graph~$\Gamma$ as the input to our construction,
we will instead take as input a bipartite
graph~$K$
which captures all of the information about~$\Gamma=(V,E)$.
The bipartite
graph~$K$
will have vertex
partition~$B = \{b_1, \ldots, b_n\}$, $T=\{t_1,\ldots,t_n\}$
and edge
set~$E(K)$
satisfying the following two properties,
\begin{enumerate}[(K1)]
\item $(b_i,t_i)\not\in
E(K)
\ \ \forall i\in [n]$,
\item $(b_i,t_j) \in
E(K) $
if and only if $(b_j,t_i)
\in E(K)$.
\end{enumerate}
The correspondence
between~$K$
and our original graph~$\Gamma$ is
as follows: We take $n$ to be $|V|$
so that there is a natural bijection between $B$ and the set $\calB= V \times \{1\}$ defined above.
Also, there is a natural bijection between $T$
and $\calT= V \times \{0\}$. The edge set
$E(K)$ is
constructed from the edge set~$E$ of~$\Gamma$ as follows.
Suppose that  $(u,1)$ is the $i$'th element of $\calB$
and that $(v,0)$ is the $j$'th element of $\calT$.
The edges $(b_i,t_j)$ and $(b_j,t_i)$ are included
in~$E(K)$
if and only if $(u,v)$ is an edge of~$E$.
The reader should verify that the
graph~$K$
encodes all of the information about the input graph~$\Gamma$,
in the sense that we could reconstruct~$\Gamma$
given~$K$.
Also, for every undirected graph~$\Gamma$,
there is a
corresponding~$K$,
and it can be constructed in polynomial time.

The sole problem remaining (and it is a big one!) is to show how,
given~$K$,
(and therefore deducing~$\Gamma$
and $\Gamma'$),
to construct a $\fourSRa$ instance~$I$ so
that its nonsingleton rotations can be labelled bijectively with
$\calB\cup \calT$ in such a way that Equations
(\ref{eq:goal0})--(\ref{eq:goal1}) are satisfied. As soon as we accomplish that, we have finished the reduction
and the proof.

The whole point of introducing the bipartite
graph~$K$
is that we
can re-state the Equations (\ref{eq:goal0})--(\ref{eq:goal1}) in a
manner that will be more convenient to work with.
In particular, using the bijection between $B$ and $\calB$ and the bijection between
$T$ and $\calT$,
these equations are equivalent to the following.

\begin{enumerate}[(G1)]
\item\label{finalgoalone}
For all $i\in[n]$,
the rotation labeled $b_i$ is dual to the rotation labeled $t_i$.
\item
\label{finalgoalthree}
The set of pairs $(b_i,t_j)$ in $\Pi$ is
$E(K)$.
\item
\label{finalgoaltwo}
There are no pairs $(t_i,b_j)$ in $\Pi$.
\item
\label{finalgoalfour}
There are no pairs   $(t_i,t_j)$ in $\Pi$ except for those with $i=j$ (which are all present).
\end{enumerate}

\subsection{The task remaining}
\label{sec:task}

We have finally defined all of the conditions that we need. At this
point, the reader should have verified that the sole problem
remaining is, given a bipartite
graph~$K$ satisfying (K1)--(K2),
we must show how to construct, in polynomial time, a
$\fourSRa$-instance $I$ satisfying (G1)--(G4). Once we do that, we
are finished with the reduction and the proof. We will not have to
further consider the graphs~$\Gamma$ and~$\Gamma'$. These were
needed only to get to this point.

 \subsection{The construction of~$I$}
 \label{sec:construct}

Our first task will be to show how to construct~$I$,
given~$K$,
and for this it will be helpful to define some notation for  describing the bipartite
graph~$K$.
Similar to the construction defined in~\cite{CGM}, label
the edges
in~$E(K)$
according to the lexicographic order of the pair
$(b_i,t_j)$ (so the edges with the smallest labels are incident
to~$b_1$). The first edge incident to $b_1$ is given the label
$(1,2)$, the second is given the label $(3,4)$ and so on. The $i$th
edge in the lexicographic ordering is given the label $(2i-1,2i)$.
Let $m$ denote the number of edges
in $K$.
The roommate instance we
construct will have $4m$ people in it, which we denote
as
$P_1,\ldots,P_{2m},Q_1,\ldots,Q_{2m}$.

We define two permutations~$\rho$ and~$\sigma$ of~$[2m]$ as
in~\cite{CGM}, so the first $\rho$-cycle corresponds to the labels
of the edges incident on~$b_1$, the first $\sigma$-cycle corresponds
to the labels of the edges incident on $t_1$, and so on. We explain
this with the help of the example in Figure \ref{fig:bip-graph1}.

\begin{figure}[ht]
\begin{center}
\begin{tikzpicture}
     [inner sep=0.5mm]

  \node (b_1) at (0,0) [circle,draw,label=below:$b_1$] {};
  \node (b_2) at (4,0) [circle,draw,label=below:$b_2$] {};
  \node (b_3) at (8,0) [circle,draw,label=below:$b_3$] {};
  \node (b_4) at (12,0) [circle,draw,label=below:$b_4$] {};

  \node (t_1) at (0,3) [circle,draw,label=above:$t_1$] {};
  \node (t_2) at (4,3) [circle,draw,label=above:$t_2$] {};
  \node (t_3) at (8,3) [circle,draw,label=above:$t_3$] {};
  \node (t_4) at (12,3) [circle,draw,label=above:$t_4$] {};

  \draw[-] (b_1) to node[very near start, above, left] {\scriptsize{$(1,2)$}} (t_2);

  \draw[-] (b_2) to node[very near start, below, left] {\scriptsize{$(3,4)$}} (t_1);
  \draw[-] (b_2) to node[near start, above, left] {\scriptsize{$(5,6)$}}(t_3);
  \draw[-] (b_2) to node[very near start, below, right] {\scriptsize{$(7,8)$}} (t_4);

  \draw[-] (b_3) to node[very near start, below, left] {\scriptsize{$(9,10)$}} (t_2);
  \draw[-] (b_3) to node[very near start, below, right] {\scriptsize{$(11,12)$}} (t_4);

  \draw[-] (b_4) to node[very near start, below, left] {\scriptsize{$(13,14)$}} (t_2);
  \draw[-] (b_4) to node[near start, above, right]{\scriptsize{$(15,16)$}} (t_3);

  \draw (-2,3) node {$T$};
  \draw (-2,0) node {$B$};

 \end{tikzpicture}\caption{Defining $\rho$-cycles and $\sigma$-cycles
 from~$K$.
 }
 \label{fig:bip-graph1}
$\rho$-cycles : (1,2), (3,4,5,6,7,8), (9,10,11,12), (13,14,15,16)\\
$\sigma$-cycles : (3,4), (1,2,9,10,13,14), (5,6,15,16), (7,8,11,12)
\end{center}
\end{figure}

\From Figure \ref{fig:bip-graph1}, it is clear that the $i$th
$\rho$-cycle is obtained by  taking together the labels of the
edges incident at node $b_i$. The first $\sigma$-cycle involves the
labels of edges incident to $t_1$. Vertex $t_1$ has only one edge
incident to it, namely edge (3,4). Hence, the first $\sigma$-cycle
is (3,4). The second $\sigma$-cycle involves edges incident to
$t_2$, namely, edges (1,2),(9,10) and (13,14). The second
$\sigma$-cycle is obtained by grouping together the labels of the
three edges and the $\sigma$-cycle is $(1,2,9,10,13,14)$. In this
manner, we obtain the remaining $\sigma$-cycles.

Note that here we have more structure than was present in the
construction from~\cite{CGM}
--- (K1) and (K2)
ensure that the number of $\rho$-cycles is equal to the number of
$\sigma$-cycles, and the number of elements in the $i$th
$\rho$-cycle is identical to the number of elements in the $i$th
$\sigma$-cycle.

Suppose there are $n$ $\rho$-cycles, and, hence, $n$
$\sigma$-cycles. Suppose also that $|\rho_i|=q_i$, i.e.\ the $i$th
$\rho$-cycle consists of $q_i$ elements.  Thus, the $i$th
$\sigma$-cycle also has $q_i$ elements and $q_i$ is even.
Note that for $1\leq k \leq q_i/2$, the elements $2k-1$ and $2k$ are in the same $\rho$-cycle and
they are in the same $\sigma$-cycle.
Also, a given $\rho$-cycle and a given $\sigma$-cycle intersect in at most one such pair.
In what follows, we let the ``representative'' of each $\rho$ cycle be the (numerically) smallest
number in the cycle.
In Figure~\ref{fig:bip-graph1}, the representatives   of the $\rho$ cycles are elements $1$, $3$, $9$ and $13$.
Note that each representative of a $\rho$-cycle is an odd number.
We let   $Rep(\rho)=\{f_1, \ldots, f_n\}$
denote the set of representatives, so, for each
$i\in \{1, \ldots,
n\}$, the
cycle $\rho_i$ can be represented as $\rho_i = (f_i, \rho f_i,
\rho^2 f_i, \ldots, \rho^{q_i-1} f_i)$.

Let $\psi:[m]\rightarrow [m]$ be the bijection defined so that, if $\ell$ is
the $k$th element of the $i$th $\rho$-cycle, then $\psi(\ell)$ is the
$k$th element of the $i$th $\sigma$-cycle.
So if the $i$'th $\rho$-cycle is $(i_1,\ldots,i_d)$ then the $i$'th $\sigma$-cycle
is $(\psi(i_1),\ldots,\psi(i_d))$.
\From the properties of the
bipartite graph $K$,
we note that $\psi$ is an involution, i.e.\
$\psi\circ\psi=\textrm{ identity}$, where
$\psi\circ\psi$ is usual function composition.
Let $Rep(\sigma)=\{\psi(f_i) \mid f_i\in Rep(\rho)\}$ be the set of representatives of the $n$ $\sigma$-cycles.

Recall that to specify our roommate instance~$I$, we must define
$4$-dimensional position
vectors $\overline{P}_i, \overline{Q}_i$ and preference vectors
$\widehat{P}_i, \widehat{Q}_i$ for each person.

Before doing this, we give a quick look ahead at what is to come in
the construction.  First of all, the instance we specify will have no
singleton rotations.   In order to satisfy condition (G\ref{finalgoalone}), the pair $(b_i,t_i)$ will
correspond to a dual pair of rotations.
The rotation associated with the vertex~$b_i$ will relate to the
$i$th $\rho$-cycle.
Suppose that this cycle is $(i_1,\ldots,i_d)$. Then the rotation will be
\begin{equation}
\begin{array}{c|cc}
Q_{i_1} & P_{i_1} P_{i_2} \\
Q_{i_2} & P_{i_2} P_{i_3} \\
\ldots & \ldots \\
Q_{i_d} & P_{i_d} P_{i_1}.\\
\end{array}
\label{eq:Q-pref-lists}
\end{equation}
Recalling Definition~\ref{def:dual-rotation}, the dual rotation,
associated with the vertex~$t_i$, is
\begin{equation}
\begin{array}{c|cc}
P_{i_1} & Q_{i_d} Q_{i_1} \\
P_{i_2} & Q_{i_1} Q_{i_2} \\
\ldots & \ldots \\
P_{i_d} & Q_{i_{d-1}} Q_{i_d}.\\
\end{array}
\label{eq:P-pref-lists}
\end{equation}
In order to satisfy condition (G\ref{finalgoalthree})
we need, for $(b_i,t_j)
\in E(K)$,
for  rotation~$b_i$ to precede rotation~$t_j$.
These dependencies will be captured
in our construction by making sure that other
 (appropriately chosen) $P$ people
are on the preference lists of the $P$ people, between the $Q$s given
on the lists in~(\ref{eq:P-pref-lists}).

The remaining parts of this section are devoted to giving
the detailed construction of our $4$-dimensional roommate instance~$I$
(which will have no singleton rotations) and showing that its nonsingleton rotations
are exactly those given in (\ref{eq:Q-pref-lists}) and (\ref{eq:P-pref-lists}).
Using this bijection between the rotations and $B\cup T$  we then show that (G1)--(G4) are satisfied,
as required in Section~\ref{sec:task}.

\subsection{Assigning position and preference vectors}\label{sect:preference-vectors}

We start by assigning the first two coordinates of the position
vectors. The people $Q_1,\ldots,Q_{2m}$ have $0$ in each of these
coordinates. The first two coordinates of the positions of
$P_1,\ldots,P_{2m}$ are arranged around a unit circle, taking each
$\rho$-cycle in order (and leaving a big gap before the next
$\rho$-cycle).
Let $\epsilon = \frac{2\pi}{(2m)^2}$.
The $i$'th $\rho$ cycle
takes up an angle of $\epsilon$ (out of the $2\pi$ radians around the circle),
starting at an angle of $2\pi(i-1)/n$.
Let $\theta_i = \epsilon/(7(q_i-1))$.
If the $i$'th  $\rho$-cycle is the cycle $(i_1,\ldots,i_d)$, then positions
$\overline{P}_{i_1},\ldots,\overline{P}_{i_d}$ are assigned in
order, leaving a gap of $7\theta_i$ between each pair of people.

More formally, we define the first two coordinates
of the position vectors as follows (the asterisks in the second
two
coordinates will be defined shortly):

\begin{align*}
\noalign
{ For $f_i \in Rep(\rho)$ ,  for $0 \leq  k \leq q_i - 1$, set}\\
  \overline{P}_{\rho^k f_i} &=  \left(\cos(2\pi(i-1)/n+7k\theta_i),\ \sin(2\pi(i-1)/n+7k\theta_i)  ,\  *,\ *\right) \ \textrm{ and }  \\
  \overline{Q}_{\rho^k f_i} &= \left(0,\ 0,\ *,\ *\right).
\end{align*}

We next assign the last two coordinates of the positions.
Once again, these coordinates are arranged around a unit circle,
taking each $\rho$-cycle in order
and leaving big gaps between consecutive $\rho$-cycles.
The $i$'th $\rho$ cycle takes up an angle of up to $\epsilon$ starting at an angle of $2\pi(i-1)/n$.
Let $\theta'_i = \epsilon/ 4$. If the $i$'th $\rho$-cycle is
  $(i_1,\ldots,i_d)$, then positions are assigned
in the following order:
 $$ \overline{Q}_{i_1} \overline{P}_{\psi(i_2)}
\overline{Q}_{i_2} \overline{P}_{\psi(i_3)}
\overline{Q}_{i_3} \overline{P}_{\psi(i_4)} \cdots
\overline{Q}_{i_{d-1}} \overline{P}_{\psi(i_d)}
\overline{Q}_{i_d} \overline{P}_{\psi(i_1)}.$$

For $k\in\{1,\ldots,d-1\}$, the
angle between $\overline{Q}_{i_{k}}$
and $\overline{Q}_{i_{k+1}}$ is
$2^{-(k-1)} 2 \theta'_i$. Also,  $\overline{P}_{\psi(i_{k+1})}$
is at an equal angle between these.

To simplify the notation while assigning the $i$'th $\rho$-cycle $(i_1,\ldots,i_d)$,
let $i_{d+1}$ denote $i_1$.
The second coordinates of the position vectors are
defined in the following manner. Note that the asterisks
representing values in the first two coordinates have   already been
defined above.

\begin{align*}
\noalign
{For $  f_i \in Rep(\rho)$ ,  for $ 0 \leq  k \leq q_i - 1$, set}\\
\overline{Q}_{\rho^k f_i} &= \Big(*,\ *,\
\cos(
\tfrac{2\pi(i-1)}{n}+2\theta'_i\sum_{j=0}^{k-1}2^{-j}), \sin(
\tfrac{2\pi(i-1)}{n}+2\theta'_i\sum_{j=0}^{k-1} 2^{-j})\Big) \ \textrm{ and } \\
\overline{P}_{\psi(\rho^{k+1} f_i)} &=  \left(*,\ *  ,\
\cos(\tfrac{2\pi(i-1)}{n}+2\theta'_i\sum_{j=0}^{k-1}2^{-j}
+2^{-k}\theta'_i),\right.\\
   &\hspace{1.2in}\left. \sin(\tfrac{2\pi(i-1)}{n}+2\theta'_i\sum_{j=0}^{k-1}2^{-j}
+2^{-k}\theta'_i)\right).
\end{align*}

A sum of the form $\sum_{j=0}^{-1}  2^{-j}$ is taken to be equal to
$0$.

Having defined the position vectors, we now give the preference
vectors.
These are also defined using the $\rho$-cycles, and again are placed around
a unit circle.

The preference vectors $\widehat{Q}_j$ are $0$ in the last two coordinates.
If the $i$'th $\rho$-cycle is $(i_1,\ldots,i_d)$, then, in the
first two coordinates,  for $1\leq j < d$, $\widehat{Q}_{i_j}$ is
placed between $\overline{P}_{i_j}$ and $\overline{P}_{i_{j+1}}$,
slightly closer to $\overline{P}_{i_j}$.
Here is the definition.
 For $f_i \in Rep(\rho)$ ,  for $0 \leq  k \leq q_i - 1$, set
 $$
\widehat{Q}_{\rho^k f_i} = \left(\cos(2\pi(i-1)/n+7k\theta_i + 3\theta_i),\
\sin(2\pi(i-1)/n+7k\theta_i + 3\theta_i)  ,\  0,\ 0 \right).  $$

Suppose that the $i$'th $\rho$-cycle is $(i_1,\ldots,i_d)$. Looking
at the preference vectors $\widehat{Q}_{i_j}$ and the position
vectors $\overline{P}_{i_j}$, we conclude that, for $j<d$, the
preference list of $Q_{i_j}$ will start $P_{i_j} P_{i_{j+1}}$. This
is consistent with our desired rotation~(\ref{eq:Q-pref-lists}). The
preference list of $Q_{i_d}$ will start $P_{i_d} \{P_{i_{d-1}}
\ldots P_{i_2}\} P_{i_1}$. We refer to the part  contained in $\{\}$
symbols as ``noise'' on the preference list.   We will have to show
that this ``noise'' does not introduce any extra rotations aside
from the ones we desire.

The preference vectors $\widehat{P}_{j}$
are $0$ in the first two coordinates.
If the $i$'th $\rho$-cycle is $(i_1,\ldots,i_d)$, then,
in the last two coordinates, for $j>1$,
$\widehat{P}_{i_j}$ is placed $1/3$ of the way along
between $\overline{Q}_{i_{j-1}}$ and $\overline{P}_{\psi(i_j)}$.
Then, $\widehat{P}_{i_1}$ is placed between
$\overline{Q}_{i_d}$ and $\overline{P}_{\psi(i_1)}$,
slightly nearer to $\overline{Q}_{i_d}$. Here is the definition.
For  $f_i \in Rep(\rho)$,  for  $0 \leq  k \leq q_i - 1$, set
$$
\widehat{P}_{\rho^{k+1} f_i} =  \Big(0,0,\ \cos(
\tfrac{2\pi(i-1)}{n}+2\theta'_i\sum_{j=0}^{k-1}2^{-j}  + \tfrac13
2^{-k} \theta'_i), \sin(
\tfrac{2\pi(i-1)}{n}+2\theta'_i\sum_{j=0}^{k-1} 2^{-j}  + \tfrac13
2^{-k} \theta'_i )\Big).
$$

Suppose that the $i$'th $\rho$-cycle is $(i_1,\ldots,i_d)$. Looking
at the preference vectors $\widehat{P}_{i_j}$ and the last two
coordinates of the position vectors, we conclude that, for $j>1$,
the   preference list of $P_{i_j}$ starts $ {Q}_{i_{j-1}}
{P}_{\psi(i_j)} Q_{i_j}$. This is consistent with the rotation
(\ref{eq:P-pref-lists}) except for the ${P}_{\psi(i_j)}$ in the
second position. This is by design, and will help to ensure the
desired precedences between the rotations. The preference list of
$P_{i_1}$ starts
$${Q}_{i_d} {P}_{\psi(i_1)}
\{
P_{\psi(i_d)}
{Q}_{i_{d-1}} \cdots
{P}_{\psi(i_3)}
{Q}_{i_2}
{P}_{\psi(i_2)} \}
Q_{i_1}\cdots,$$
where, once again, the  part of the list in $\{\}$ is ``noise'', which is not
desired, but will turn out to do no harm.

Having given our construction, we have to show that the rotations that get
exposed are exactly
the ones that we have identified in (\ref{eq:Q-pref-lists}) and (\ref{eq:P-pref-lists}).
Using this bijection between the rotations and $B\cup T$  we will then show that (G1)--(G4) are satisfied,
as required in Section~\ref{sec:task}.

We start by summarising the observations that we have made about the preference lists.
Consider one of the $\rho$-cycles, $\rho_i=(i_1, \ldots, i_d)$.
Using our position and preference vectors, we
can write down prefixes of   the preference lists   for the people $P_{i_1},\ldots,P_{i_d},Q_{i_1},\ldots,Q_{i_d}$ in this
$\rho$-cycle.
\begin{equation}
\begin{array}{c|ll}
 Q_{i_1} & P_{i_1} P_{i_2} \ \cdots \\
 Q_{i_2} & P_{i_2} P_{i_3} \ \cdots \\
 \ldots & \ldots \\
 Q_{i_{d-1}} & P_{i_{d-1}} P_{i_d} \ \cdots\\
 Q_{i_d} & P_{i_d} \{
P_{i_{d-1}} \cdots P_{i_2}
\}  P_{i_1} \ \cdots \\
 \end{array}
\label{eq:noisy-Q-pref-lists}
\end{equation}

\begin{equation}
\begin{array}{c|ll}
 P_{i_1} & Q_{i_d} {P}_{\psi(i_1)}
 \{
P_{\psi(i_d)}
{Q}_{i_{d-1}} \cdots
{P}_{\psi(i_3)}
{Q}_{i_2}
{P}_{\psi(i_2)} \} Q_{i_1} \ \cdots \\
 P_{i_2} & Q_{i_1} {P}_{\psi(i_2)} Q_{i_2} \ \cdots \\
 \ldots & \ldots \\
 P_{i_{d-1}} & Q_{i_{d-2}} {P}_{\psi(i_{d-1})} Q_{i_{d-1}} \ \cdots \\
 P_{i_d} & Q_{i_{d-1}} {P}_{\psi(i_d)} Q_{i_d}\ \cdots \\
 \end{array}
\label{eq:noisy-P-pref-lists}
\end{equation}

We have not given the entire preference lists, but only the parts that
are relevant for us.  As we show in the next section, the portion
of each preference list that is given above is the only part that remains after
Phase 1 of Irving's Roommate Algorithm.  In fact, additional
parts of some preference lists, specifically the parts listed in
braces, vanish during the execution of Phase 1.

\begin{example} The prefixes of the preference lists for the example
from Figure~\ref{fig:bip-graph1} are given in the second appendix.
\end{example}

 \begin{remark}
Strictly speaking, the construction that we have given does not fit
into our original definition of a roommate instance.  In particular,
with the position and preference vectors we have defined, we note
that $\widehat{Q}_i\cdot \overline{Q}_j = 0$ for all $i\not= j$.
This means that each $Q_i$ has a tie in his/her preference for all
of the other $Q_j$s ($j\not= i$).  However, it is easy to modify the
position vectors $\overline{Q}_j$ so that we have a strict
preference ordering for each person.  We can also do this in such a
way that the {\em beginning} of the preference lists of the $Q_i$s
is undisturbed.  We may, for example, pick some very small $\delta_j
> 0$ and assign the first two coordinates of $\overline{Q}_j$ to
equal $\delta_j$, resulting in a strict preference list for $Q_i$.
By choosing the $\delta_j$s small enough, we can get a strict
preference for all of the people without altering the beginning
segment of the preference list of each person.  In particular, as
long as the start of the preference lists given
in~(\ref{eq:noisy-Q-pref-lists}) and~(\ref{eq:noisy-P-pref-lists})
is maintained, this is sufficient for our purposes.
\end{remark}

\subsection{Enumerating rotations}\label{sect:finding-rotations}

In this section, we will establish the rotations in the stable
roommate instance $I$.
First, consider Phase 1. From (\ref{eq:noisy-Q-pref-lists}) and (\ref{eq:noisy-P-pref-lists})
 we see that, during Phase~1, each $Q_{i_j}$ will
become semi-engaged to $P_{i_j}$ and each $P_{i_j}$ will become
semi-engaged to $Q_{i_{j-1}}$.
Since the outcome of Phase~1 is independent of the order in which free people make proposals, we can
assume that these proposals occur in order.
Thus, the  suffixes of the preference lists that are omitted
from (\ref{eq:noisy-Q-pref-lists}) and (\ref{eq:noisy-P-pref-lists})
will disappear by the end of Phase~1. The purpose of this section is to show that, by the end of Phase 1,
the preference lists look like (\ref{eq:unnoisy-Q-pref-lists}) and (\ref{eq:unnoisy-P-pref-lists}).

By construction, the length of each cycle is even.
If $d=2$ then the preference list of $Q_{i_d}$ in (\ref{eq:noisy-Q-pref-lists})
has no noise. Otherwise,
$P_{i_2},\ldots,P_{i_{d-1}}$ do not have $Q_{i_d}$
on their short lists in (\ref{eq:noisy-P-pref-lists})
so, since the process of removing people from each other's lists is symmetric,
we can conclude that, by the end of Phase~1, these people are removed from the short list of $Q_{i_d}$.

If $d=2$ then the preference list of $P_{i_1}$
has prefix $Q_{i_2} P_{\psi(i_1)}  P_{\psi(i_2)}  Q_{i_1}$.
Since $i_1$ is the representative of the $\rho$-cycle $(i_1,i_2)$, it is odd,
so $i_2$ is even and therefore $\psi(i_2)$ is even, and is therefore not the representative
of any $\rho$-cycle.
Thus, the preference list of $P_{\psi(i_2)} = P_{\psi(i_d)}$
starts with $
Q_{\rho^{-1}(\psi(i_d))}
P_{\psi^2(i_d)} Q_{\psi(i_d)} = Q_{\rho^{-1}(\psi(i_d))} P_{i_d}
Q_{\psi(i_d)}$.
$P_{i_1}$ comes after this prefix on the preference list.
Thus, $P_{i_1}$ is removed from this preference list by the end of Phase~1.
Symmetrically, $P_{\psi(i_d)}$ is   removed from the preference list of $P_{i_1}$ by the end of Phase~1.

Now suppose $d>2$.
Suppose that $(i_1,\ldots,i_d)$ is the $k$'th $\rho$-cycle.
Now
$Q_{i_2},\ldots,Q_{i_{d-1}}$ do not have $P_{i_1}$ on their short lists
in (\ref{eq:noisy-Q-pref-lists}) so we can conclude that, by the end of Phase~1, these people
are removed from the short list of $P_{i_1}$.
For $\ell\in\{2,\ldots,d\}$, consider
the person $P_{\psi(i_\ell)}$, which is on the short list of $P_{i_1}$ in
(\ref{eq:noisy-P-pref-lists}). If $\psi(i_\ell)$ is not the representative of a $\rho$-cycle,
then, using the same argument that
we used in the $d=2$ case, the preference list of $P_{\psi(i_\ell)}$
starts with
$ Q_{\rho^{-1}(\psi(i_\ell))} P_{i_\ell}
Q_{\psi(i_\ell)}$, so  $P_{i_1}$ follows this prefix and $P_{\psi(i_\ell)}$ is removed from the preference list of $P_{i_1}$ by the end of Phase~1.
Suppose that
$\psi(i_\ell)$ is the representative of a $\rho$-cycle, say the $t$'th $\rho$-cycle.
To ease notation, suppose that $\psi(i_\ell) = j_1$ and that $\rho_t = (j_1,\ldots,j_{d'})$.
\From (\ref{eq:noisy-P-pref-lists}),
we know that the preference list of~$P_{j_1}$ starts with the following prefix
\begin{equation}\label{eq:temptemp}   Q_{j_{d'}} {P}_{\psi(j_1)}
 \{
P_{\psi(j_{d'})} {Q}_{j_{d'-1}} \cdots {P}_{\psi(j_3)} {Q}_{j_2}
{P}_{\psi(j_2)} \} Q_{j_1}.\end{equation} Then $\psi(j_1)= i_\ell$
is a member of the $t$-th $\sigma$-cycle and the $k$'th
$\rho$-cycle. $ \psi(i_\ell)$ is odd, so $i_\ell$ is odd, so
$i_\ell$ and $i_\ell+1 (= i_{\ell + 1})$ are the only two elements
that are in both of these cycles. We conclude that $i_1$ is not in
both of these cycles, so it is not in the $t$'th $\sigma$-cycle.
Thus, none of the people $P_{\psi(j_{d'})}
  \cdots
{P}_{\psi(j_3)}
 {P}_{\psi(j_2)}$ on the prefix (\ref{eq:temptemp}), all of whom are in the $t$'th $\sigma$-cycle,
 is equal to $P_{i_1}$. Therefore, $P_{i_1}$ comes after the prefix depicted in (\ref{eq:temptemp}).
As above, we can conclude that
$P_{\psi(i_\ell)}$ is removed from the preference list of $P_{i_1}$ by the end of Phase~1.

Hence, after Phase~1 the  short lists look like this.
\begin{equation}
\begin{array}{c|ll}
 Q_{i_1} & P_{i_1} P_{i_2} \\
 Q_{i_2} & P_{i_2} P_{i_3} \\
 \ldots & \ldots \\
 Q_{i_{d-1}} & P_{i_{d-1}} P_{i_d}\\
 Q_{i_d} & P_{i_d}   P_{i_1}\\
 \end{array}
 \label{eq:unnoisy-Q-pref-lists}
 \end{equation}
\begin{equation}
 \begin{array}{c|ll}
 P_{i_1} & Q_{i_d} {P}_{\psi(i_1)} Q_{i_1} \\
 P_{i_2} & Q_{i_1} {P}_{\psi(i_2)} Q_{i_2} \\
 \ldots & \ldots \\
 P_{i_{d-1}} & Q_{i_{d-2}} {P}_{\psi(i_{d-1})} Q_{i_{d-1}}\\
 P_{i_d} & Q_{i_{d-1}} {P}_{\psi(i_d)} Q_{i_d}\\
\end{array}
\label{eq:unnoisy-P-pref-lists}
\end{equation}

\begin{example}
The short lists for the example from Figure~\ref{fig:bip-graph1} are given in the second appendix.
\end{example}

After Phase~1, we know that for each $\rho$-cycle,
each $Q_{i_j}$ is semi-engaged to $P_{i_j}$ and each $P_{i_j}$ is
semi-engaged to $Q_{i_j-1}$.
According to Section~\ref{sec:task} we now need to identify the rotations, and   to establish
conditions (G\ref{finalgoalone})--(G\ref{finalgoalfour}).
Then we are finished.

\begin{lemma}\label{lem:bottom-rotations}
Suppose the $n$ $\rho$-cycles of the bipartite
graph $K$ are
$\rho_1, \rho_2, \ldots, \rho_n$, where
$\rho_j = (i_{j-1}+1,\ldots, i_{j})$ for $j \in [n]$, with $i_0 = 0$ and $i_{n} = 2m$.

The shortlists obtained after Phase~1 have exactly $n$ exposed rotations,
$R_1$ through $R_n$, where $R_j = (E_j, H_j, S_j)$,
\begin{eqnarray*}
  E_j & = & \{Q_{i_{j-1}+1}, Q_{i_{j-1}+2}, \ldots, Q_{i_j}\}, \\
  H_j & = & \{P_{i_{j-1}+1}, P_{i_{j-1}+2}, \ldots, P_{i_j}\}, \textrm{ and } \\
  S_j & = & \{P_{i_{j-1}+2}, P_{i_{j-1}+3}, \ldots, P_{i_j}, P_{i_{j-1}+1}\}.
\end{eqnarray*}
\end{lemma}

\begin{proof}
By construction,
$Rep(\rho) = \{i_0+1, i_1+1,\cdots,i_{n-1}+1\}$.
\From (\ref{eq:unnoisy-Q-pref-lists}) and (\ref{eq:unnoisy-P-pref-lists}), the short lists of the $Q_{*}$ and the
$P_{*}$ people after Phase~1 can be summarized as follows:
\begin{align*}
\textrm{For}\ j &\in Rep(\sigma), \\
\nonumber Q_{\rho^k j} \;\;&:\;\;\; P_{\rho^k j} P_{\rho^{k+1}
  j}\;\;,\;\;\; 0\leq k \leq q_j-1, \\
\nonumber P_{\rho^{k} j} \;\;&:\;\;\; Q_{\rho^{(k-1)} j}
P_{\psi(\rho^{k} j)}Q_{\rho^{k} j}\;\;,\;\;\; 0\leq k \leq q_j-1.
\end{align*}

We observe that the $Q_{*}$ people whose indices belong to
$\rho_j$, along with their first and second preferences, form an exposed
rotation $R_j = (E_j, H_j, S_j)$, where $E_j, H_j$, and $S_j$ are
as defined in the statement of the lemma.
This is   a rotation as displayed in (\ref{eq:Q-pref-lists}).
Therefore, at the end of
Phase~1, the stable roommate instance $I$ has {\em at least} $n$ exposed
rotations.  Further, each $Q_{*}$ person appears in one of these $n$ rotations,
meaning that each $Q_{*}$ appears in one of the sets $E_j$.

Next, we show that the above $n$ rotations are the {\em only} rotations
that are exposed after Phase~1. Suppose person $P_i$ is in an
exposed rotation $R =(E, H, S)$, i.e.\ $P_i \in E$. The initial part
of the preference list of $P_i$ at the end of Phase~1 is
$Q_{\rho^{-1} i} P_{\psi(i)}$.  Hence, by definition, the next person
in the ordered set $E$ is $Q_{\psi(i)}$ who is currently
semi-engaged to $P_{\psi(i)}$. Suppose $\psi(i)$ belongs to the $k$th
$\rho$-cycle $\rho_k$. We know that $Q_{\psi(i)}$ is part of the exposed
rotation $R_k = (E_k, H_k, S_k)$, i.e.\ $Q_{\psi(i)} \in E_k = \{Q_j
: j \in \rho_k\}$. This implies that $P_i \in E_k$ which is not
possible.

Hence, if $R = (E, H, S)$ is an exposed rotation in the
table obtained after Phase~1, then $P_i \notin E$ for all $i$.
Therefore, there are exactly $n$ exposed rotations $R_1, R_2,
\cdots, R_n$ at the end of Phase~1.
\end{proof}

\begin{example}
For the example  from Figure~\ref{fig:bip-graph1}, there are $4$ exposed rotations,
$R_1=\{E_1,H_1,S_1\}$,
$R_2=\{E_2,H_2,S_2\}$,
$R_3=\{E_3,H_3,S_3\}$ and
$R_4=\{E_4,H_4,S_4\}$.  These are
given as follows.
\begin{eqnarray*}
  E_1 & = & \{Q_1,Q_2\}, \\
  H_1 & = & \{ P_1,P_2\}, \\
  S_1 & = & \{ P_2, P_1\}, \\
  E_2 & = & \{Q_3, Q_4, Q_5, Q_6, Q_7, Q_8\}, \\
  H_2 & = & \{P_3, P_4, P_5, P_6, P_7, P_8\}, \\
  S_2  & = & \{ P_4, P_5, P_6, P_7, P_8, P_3\},\\
  E_3 & = & \{Q_9, Q_{10}, Q_{11}, Q_{12} \}, \\
  H_3 & = & \{ P_9, P_{10}, P_{11}, P_{12}\}, \\
  S_3  & = & \{ P_{10}, P_{11}, P_{12}, P_9 \},\\
    E_4 & = & \{Q_{13}, Q_{14}, Q_{15}, Q_{16} \}, \\
  H_4 & = & \{ P_{13}, P_{14}, P_{15}, P_{16} \}, \\
  S_4  & = & \{  P_{14}, P_{15}, P_{16}, P_{13} \}.\\
    \end{eqnarray*}

\end{example}

Now we show that each of the $n$ rotations from Lemma~\ref{lem:bottom-rotations}
has a dual rotation of the form~(\ref{eq:P-pref-lists}).

\begin{lemma}\label{lem:top-rotations}
For $j \in [n]$, there is a dual rotation $R^d_{j} = \{E^d_{j},
H^d_{j}, S^d_{j}\}$ corresponding to rotation $R_j$ from Lemma~\ref{lem:bottom-rotations},
where

\begin{align*}
  E^d_j & =  S_j = \{P_{i_{j-1}+2}, P_{i_{j-1}+3}, \ldots, P_{i_j}, P_{i_{j-1}+1}\}, \\
  H^d_j & =  E_j = \{Q_{i_{j-1}+1}, Q_{i_{j-1}+2}, \ldots, Q_{i_j}\}, \\
  S^d_j & = E_j^r =
  \{Q_{i_{j-1}+2}, \ldots, Q_{i_j},  Q_{i_{j-1}+1}\}. \\
\end{align*}
\end{lemma}

\begin{proof}
First we note that $R_j^d$ as defined above has the {\em form} of a
rotation and will be the dual to $R_j$ {\em provided}  that  there
is some sequence of rotations that can be performed that leads to
$R_j^d$ being exposed in the resulting table. We show that there is
such a sequence of rotations.

Let
$\rho_j = (i_{j-1}+1, i_{j-1}+2, \cdots, i_{j})$
denote the $\rho$-cycle
corresponding to the
rotation $R_j$, where $q_j = i_j - i_{j-1} = |\rho_j|$.
The $j$'th $\sigma$-cycle, $\sigma_j$,
also has $|\sigma_j| = q_j$,
so write $\sigma_j = (\ell_1, \ldots, \ell_{q_j}) =
(\psi(i_{j-1}+1), \psi(i_{j-1}+2), \cdots, \psi(i_{j}))$.
Recall that
$\rho_j \cap \sigma_j = \emptyset$
and that, for any $r\in[n]$,
$|\rho_j \cap \sigma_r| \in \{0,2\}$.

Given $\sigma_j$,
there are {\em exactly} $q_j/2$ distinct  $\rho$-cycles
$\rho_{t_1}, \ldots, \rho_{t_{q_j/2}}$
such that, for $k\in\{1,\ldots,q_j/2\}$,
$\ell_{2k-1}\in \rho_{t_k}$ and
$\ell_{2k} \in \rho_{t_k}$.

Consider the table $T$ obtained at the end of Phase 1. We claim that
the elimination of rotations $R_{t_1}, \ldots, R_{t_{q_j/2}}$ from
$T$, exposes the (proposed) dual rotation $R_j^d$. After the
elimination of those rotations, every $Q_r$ and $P_r$, for $r\in
\rho_{t_1}\cup \rho_{t_2}\cup\cdots\cup\rho_{t_{q_j/2}}$, will have
only one person on  his resulting short list.

Since, for $k\in\{1,\ldots,q_j/2\}$, $\ell_{2k-1}$ and $\ell_{2k}$ are in $\rho_{t_k}$,
$P_{\ell_{2k-1}}$ has $P_{\psi(\ell_{2k-1})}$  on his short-list prior to the elimination, but not subsequently.
Similarly, $P_{\ell_{2k}}$ has $P_{\psi(\ell_{2k})}$ on his short-list prior to the elimination, but not subsequently.
 Thus, the elimination removes
 $P_{\ell_{2k-1}}=P_{\psi(i_{j-1}+2k-1)}$ from the short list of  $P_{i_{j-1} + 2k-1} = P_{\psi(\ell_{2k-1})}$
 and it removes
 $P_{\ell_{2k}}=P_{\psi(i_{j-1}+2k)}$ from the short list of  $P_{i_{j-1}+2k} = P_{\psi(\ell_{2k})}$.

After eliminating all of the rotations $R_{t_1}, \ldots, R_{t_{q_j/2}}$,
it follows that the
preference list of $P_r$, for $r\in\rho_j$, will be $Q_{\rho^{-1}r}Q_r$.
This means  that $R_j^d$ is exposed in the resulting table, showing
that $R_j$ has a dual rotation as desired.
\end{proof}

\begin{example}
The dual
rotations
for the example  from Figure~\ref{fig:bip-graph1} are given as follows.
$R^d_1=\{E^d_1,H^d_1,S^d_1\}$,
$R^d_2=\{E^d_2,H^d_2,S^d_2\}$,
$R^d_3=\{E^d_3,H^d_3,S^d_3\}$ and
$R^d_4=\{E^d_4,H^d_4,S^d_4\}$.  These are
given as follows.
\begin{eqnarray*}
  E^d_1 & = & \{ P_2, P_1\} , \\
  H^d_1 & = &  \{Q_1,Q_2\}, \\
  S^d_1 & = &   \{Q_2, Q_1\},\\
  E^d_2 & = &   \{ P_4, P_5, P_6, P_7, P_8, P_1\}, \\
  H^d_2 & = &   \{Q_3, Q_4, Q_5, Q_6, Q_7, Q_8\},\\
  S^d_2  & = &  \{Q_4, Q_5, Q_6, Q_7, Q_8,Q_3\},\\
  E^d_3 & = &  \{ P_{10}, P_{11}, P_{12}, P_9 \}, \\
  H^d_3 & = &  \{Q_9, Q_{10}, Q_{11}, Q_{12} \}, \\
  S^d_3  & = & \{  Q_{10}, Q_{11}, Q_{12},Q_9 \}, \\
  E^d_4 & = &  \{  P_{14}, P_{15}, P_{16}, P_{13} \}, \\
  H^d_4 & = &   \{Q_{13}, Q_{14}, Q_{15}, Q_{16} \}, \\
  S^d_4  & = &  \{Q_{14}, Q_{15}, Q_{16} ,Q_{13}\}.\\
    \end{eqnarray*}

\end{example}

We state one more structural property of the rotation poset, previously
proved in~\cite{gusfield-structure}.

\begin{lemma}~\cite[Lemma 3.5]{gusfield-structure}\label{gus-RR'}
If $R = (E, H, S)$ and $R' = (E', H', S')$ are two distinct
rotations exposed in a table $T$, then $R$ removes $R'$ from $T$ if
and only if $R' = R^d$. Hence, the only way to remove an exposed
rotation is to explicitly eliminate it or its dual rotation, if it
has one.
\end{lemma}

\begin{lemma}
The only rotations associated with the stable roommate instance $I$
are $R_j$ and $R^d_j$ for $j \in [n]$.
\label{lem:allrotations}
\end{lemma}
\begin{proof}
Suppose $R =(E, H, S)$ is a rotation different from $R_i$ and
$R^d_i$ for all $i \in [n]$, and $R$ is exposed in a table $T$.
Suppose $Q_j \in E$. This, in turn, implies that the preference list
of $Q_j$, which has at least two persons on it, is $P_j P_{\rho j}$.
Suppose $j \in \rho_k$. This means that rotation $R_k$ has not been
removed and is exposed in table $T$. Eliminating rotation $R$ would
force $Q_j$ to be semi-engaged to $P_{\rho j}$ and remove rotation
$R_k$. From Lemma~\ref{gus-RR'}, it follows that the only rotation
that could remove the exposed rotation $R_k$ from table $T$ is
$R^d_k$. This implies $R = R^d_k$ which contradicts our assumption
that $R$ is different from $R_i$ and $R^d_i$ for $i \in [n]$. Hence,
$Q_j \notin E$ for $j \in [2m]$. This establishes that the only
persons that could potentially belong to $E$ are $P_{*}$
persons.

Suppose $P_j \in E$ and the preference list of $P_j$ in table $T$
does not start with $Q_{\rho^{-1} j}$. Let $j \in \rho_k$. This entails
that $P_j$ does not belong to the preference list of $Q_{\rho^{-1}
j}$. This indicates that rotation $R_k$ has been removed as $R_k$ is
exposed at the end of Phase~1. By
Lemma~\ref{gus-RR'}, it follows that the exposed rotation $R_k$ was
removed by eliminating rotation $R^d_k$. This forces $P_j$ to be
semi-engaged to $Q_j$ and $P_j$ has only $Q_j$ on his/her list. This
rules out the possibility of $P_j \in E$.

Hence, every $P_l \in E$
has to start with $Q_{\rho^{-1} l}$, i.e.\ every $P_l \in E$ is
semi-engaged to $Q_{\rho^{-1} l}$. The preference list of $P_j$
could read either $Q_{\rho^{-1} j} P_{\psi(j)}Q_j$ or
$Q_{\rho^{-1} j} Q_j$.
Suppose the preference list of $P_j$ starts with $Q_{\rho^{-1} j}
P_{\psi(j)}$. This implies that the next person in the ordered set $E$ is
the person semi-engaged to $P_{\psi(j)}$. This is not possible since every
person that belongs to $E$ is a $P_{\cdot}$ person and every
$P_l \in E$ is semi-engaged to $Q_{\rho^{-1} l}$. Therefore, the
preference list of every $P_l \in E$ reads $Q_{\rho^{-1} l}Q_l$.
Eliminating rotation $R$ from table $T$ would result in every $P_l
\in E$ being semi-engaged to $Q_{l}$ and would remove $P_l$ from
$Q_{\rho^{-1} l}$'s list. Suppose $P_i \in E$ and $i \in \rho_k$. This
implies that $P_i$ is removed from $Q_{\rho^{-1} i}$'s list and the
exposed rotation $R_k$ is removed by rotation $R$. Again, by
Lemma~\ref{gus-RR'}, it follows that $R = R^d_k$ which is not possible.

Hence, the only rotations that the stable roommate instance $I$ has
are $R_j$ and $R^d_j$ for $j \in [n]$.
\end{proof}

\subsection{Ordering Rotations}\label{sect:ordering-rotations}

In this section we will order the rotations using the {\em
explicitly precedes} relation. We state another lemma from
\cite{gusfield-structure}.

\begin{lemma}\cite[Lemma 5.2]{gusfield-structure}\label{gus-unique R'}
Let $p$ be a person who must be removed from $e_i$'s list before
rotation $R$ is exposed. There exists a unique rotation $R'$ whose
elimination removes $p$ from $e_i$'s list.
\end{lemma}

First we establish one fact about the partial order $\Pi^*$ for
our constructed instance $I$.

\begin{lemma}
Rotations $R_1,\ldots, R_n$ are minimal elements in the partial
order $\Pi^*$.
\label{lem:aremin}
\end{lemma}

\begin{proof}
We observe that the table obtained at the end of Phase~1 has
rotations $R_1$ through $R_n$ exposed. Hence, there is no rotation
$R$ that explicitly precedes $R_i$ for $i \in [n]$. Therefore, $R_1$
through $R_n$ are minimal elements in the partial order $\Pi^*$.
\end{proof}

\begin{lemma}\label{lem:exp-prec}
The only rotations that explicitly precede rotation $R^d_j$ are
$${\mathcal R} = \{R_{t_k} : \rho_{t_k} \cap \sigma_j = \{\psi(i_{j-1}+2k-1),\psi(i_{j-1}+2k)\},
1 \leq k \leq q_j/2 \},$$ where
$\sigma_j$ is the $j$-th $\sigma$-cycle
$\{\psi(i_{j-1}+1), \psi(i_{j-1}+2), \cdots, \psi(i_{j})\}$ and
$\rho_{t_k}$ is the $t_k$-th $\rho$-cycle $\{i_{t_k-1}+1, i_{t_k-1}+2,
\cdots, i_{t_k-1}+q_{t_k}\}$ for $ 1\leq k \leq q_j/2$.

\end{lemma}

\begin{proof}
In rotation $R^d_j = (E^d_j, H^d_j, S^d_j)$, we have
\begin{align*}
  E^d_j &   = \{P_{i_{j-1}+2}, P_{i_{j-1}+3}, \ldots, P_{i_j}, P_{i_{j-1}+1}\}, \\
  H^d_j &   = \{Q_{i_{j-1}+1}, Q_{i_{j-1}+2}, \ldots, Q_{i_j}\}, \\
  S^d_j &   =
  \{Q_{i_{j-1}+2}, \ldots, Q_{i_j},  Q_{i_{j-1}+1}\}. \\
\end{align*}
The
preference lists of members of $E^d_j$   at the end of
Phase~1 are as follows for $1\leq k \leq q_j/2$:
\begin{align*}
 P_{i_{j-1}+ 2k-1} &:\quad  Q_{\rho^{-1}(i_{j-1}+ 2k-1)}P_{\psi(i_{j-1}+ 2k-1)}Q_{i_{j-1}+ 2k-1},\\
 P_{i_{j-1}+ 2k} &:\quad  Q_{\rho^{-1}(i_{j-1}+ 2k)}P_{\psi(i_{j-1}+ 2k)}Q_{i_{j-1}+ 2k}.
\end{align*}
For rotation $R^d_j$ to be exposed, the following has to occur for all $1\leq k \leq q_j/2$.
Person $P_{\psi(i_{j-1}+ 2k-1)}$ needs to be removed from the list of person
$P_{i_{j-1}+ 2k-1}$ and
person $P_{\psi(i_{j-1}+ 2k)}$ needs to be removed from the list of person
$P_{i_{j-1}+ 2k}$.

Note that after the elimination of
$R_{t_k}$, for $\ell\in \rho_{t_k}$,
$Q_\ell$ is semi-engaged
to $P_{\rho \ell}$
and $P_{\rho \ell}$ being semi-engaged to $Q_\ell$. At this point, $Q_\ell$ and $P_{\rho \ell}$ have only
one person on their lists.
In particular,
$P_{\psi(i_{j-1}+2k-1)}$ has
only $Q_{\rho^{-1}(\psi(i_{j-1}+2k-1))}$ on his/her list and is removed
from $P_{i_{j-1}+2k-1}$'s list.  Also, $P_{\psi(i_{j-1}+2k)}$ has
only $Q_{\rho^{-1}(\psi(i_{j-1}+2k))}$ on his/her list and is removed
from $P_{i_{j-1}+2k}$'s list.

\From Lemma~\ref{gus-unique R'}, it
follows that $R_{t_k}$ is the unique rotation that  does this job. Therefore, the elimination of rotations
$R \in {\mathcal R}$ exposes rotation $R^d_j$. Since every rotation
$R \in {\mathcal R}$ explicitly precedes $R^d_j$ and the elimination
of rotations in ${\mathcal R}$ exposes $R^d_j$, we conclude that the
only rotations that explicitly precede rotation $R^d_j$ are $R \in
{\mathcal R}$.
\end{proof}

\begin{example}
For the example  from Figure~\ref{fig:bip-graph1},
the rotations $R_1$, $R_2$, $R_3$, $R_4$ are minimal elements in the partial
order $\Pi^*$.
The reader can think about these as corresponding to the vertices $b_1$, $b_2$, $b_3$,
and $b_4$, respectively, in Figure~\ref{fig:bip-graph1}.
The only other rotations are $R^d_1$, $R^d_2$, $R^d_3$ and $R^d_4$,
which correspond to vertices~$t_1$, $t_2$, $t_3$ and $t_4$, respectively.
\From Lemma~\ref{lem:exp-prec}, it can be deduced that
$R_i$ explicitly precedes $R^d_j$ if and only if there is an edge from $b_i$ to $t_j$ in the figure.
No other rotations explicitly precede $R^d_j$.

\end{example}

\subsection{Stocktaking}\label{sect:Stocktaking}

\From Section~\ref{sec:task}, the task was, given the bipartite
graph~$K$
as described by the $\rho$ and $\sigma$ cycles, to
construct, in polynomial time, a $\fourSRa$-instance $I$ whose
nonsingleton rotations can be labelled bijectively with $B\cup T$ so
that (G1)--(G4) are satisfied. The construction was done in
Sections~\ref{sec:construct} and \ref{sect:preference-vectors}. The
bijection between rotations and $B\cup T$ is  given informally by
(\ref{eq:Q-pref-lists}) and (\ref{eq:P-pref-lists}). More formally,
the vertex in~$B$ corresponding to the $i$'th $\rho$-cycle is
labelled with the rotation~$R_i$ and the vertex in~$T$
corresponding to the $i$'th $\sigma$-cycle is labelled with the
rotation~$R^d_i$. Lemma~\ref{lem:allrotations} guarantees that these
are all of the rotations associated with~$I$ so the duality
relationship gives us (G1). Lemma~\ref{lem:exp-prec} gives us (G2)
and (G4). Lemma~\ref{lem:aremin} gives us (G3).

Thus, we have completed the task, and, by Section~\ref{sec:task}, we have completed the reduction
$\IS \APred \fourSRa$
and the proof that the reduction is correct.

\section{Euclidean model}

Recall that $\kSRe$ denotes the problem of counting stable
assignments for $k$-Euclidian stable roommate instances. The goal of
this section is to prove Theorem~\ref{thm:3SRe} which we restate
below.

\setcounter{counter:save}{\value{theorem}}
\setcounter{theorem}{\value{counter:SAT-3SRe}}
\begin{theorem}
$\IS \APeq \kSRe$ for $k\geq 3$.
\end{theorem}
\setcounter{theorem}{\value{counter:save}}

\begin{proof}
As in the proof of Theorem~\ref{thm:4-SR}, the main task is to prove
$\IS\APred\;\;\;$ $\threeSRe$. Given an instance of $\IS$, we show
how to construct an instance of $\threeSRe$ whose preference lists
are identical to those from
Section~\ref{section:begin-construction}.

We start by assigning position and preference points.
We describe the instance of $\IS$ in terms of $n$ $\rho$-cycles and
$n$ $\sigma$-cycles as in Section~\ref{sec:construct}. As in that section,
 we have $4m$ people whom we label $P_1$
through $P_{2m}$ and $Q_1$ through $Q_{2m}$. We start by assigning
the third co-ordinate of the $P_*$ and the $Q_*$ people. The people
$Q_1,\ldots,Q_{2m}$ have $0$ in their third coordinate. The third
coordinate of the position points of $P_1,\ldots,P_{2m}$ are
arranged on the $z$-axis, taking each $\rho$-cycle in order (and
leaving a big gap before the next $\rho$-cycle).

Let $\epsilon = \frac{1}{(2m)^2}$.
The $i$'th $\rho$ cycle
takes up  a distance of $\epsilon$ (out of  a unit distance on the $z$ axis),
starting at  distance $(i-1)/n$.
Let $\theta_i = \epsilon/(7(q_i-1))$.
If the $i$'th  $\rho$-cycle is the cycle $(i_1,\ldots,i_d)$, then positions
$\overline{P}_{i_1},\ldots,\overline{P}_{i_d}$ are assigned in
order, leaving a gap of $7\theta_i$ between each pair of people.

More formally, we define the third coordinate of the position points
as follows (the asterisks in the first and the second coordinates
will be defined shortly):

\begin{align*}
\noalign
{ For $f_i \in Rep(\rho)$ ,  for $0 \leq  k \leq q_i - 1$, set}\\
  \overline{P}_{\rho^k f_i} &=  \left(*,\ *,   (i-1)/n+7k\theta_i\right)  \ \textrm{ and }  \\
  \overline{Q}_{\rho^k f_i} &= \left(*,\ *, 0\right).
\end{align*}

We next assign the first and the second coordinates of the position
points.
These are similar to the last two coordinates in Section~\ref{section:begin-construction}. These coordinates are arranged around a circle of radius
$R$, where $R$ will be specified later, taking each $\rho$-cycle in
order and leaving big gaps between consecutive $\rho$-cycles.
The $i$'th $\rho$ cycle takes up an angle of up to $\epsilon$ starting at an angle of $2\pi(i-1)/n$.
Let $\theta'_i = \epsilon/ 4$. If the $i$'th $\rho$-cycle is
  $(i_1,\ldots,i_d)$, then positions are assigned
in the following order:
$$ \overline{Q}_{i_1} \overline{P}_{\psi(i_2)}
\overline{Q}_{i_2} \overline{P}_{\psi(i_3)} \overline{Q}_{i_3}
\overline{P}_{\psi(i_4)} \cdots \overline{Q}_{i_{d-1}}
\overline{P}_{\psi(i_d)} \overline{Q}_{i_d}
\overline{P}_{\psi(i_1)}.$$

For $k\in\{1,\ldots,d-1\}$, the
angle between $\overline{Q}_{i_{k}}$
and $\overline{Q}_{i_{k+1}}$ is
$2^{-(k-1)} 2 \theta'_i$. Also,  $\overline{P}_{\psi(i_{k+1})}$
is at an equal angle between these.

To simplify the notation while assigning the $i$'th $\rho$-cycle $(i_1,\ldots,i_d)$,
let $i_{d+1}$ denote $i_1$.
The  first two coordinates of the position  points are
defined in the following manner. Note  that the asterisk
representing the value of the third  coordinate has   already been
defined above.

\begin{align*}
\noalign
{For $  f_i \in Rep(\rho)$ ,  for $ 0 \leq  k \leq q_i - 1$, set}\\
\overline{Q}_{\rho^k f_i} &= \Big( R
\cos(
\tfrac{2\pi(i-1)}{n}+2\theta'_i\sum_{j=0}^{k-1}2^{-j}), R \sin(
\tfrac{2\pi(i-1)}{n}+2\theta'_i\sum_{j=0}^{k-1} 2^{-j}),*\Big) \ \textrm{ and } \\
\overline{P}_{\psi(\rho^{k+1} f_i)} &=  \left(   R
\cos(\tfrac{2\pi(i-1)}{n}+2\theta'_i\sum_{j=0}^{k-1}2^{-j}
+2^{-k}\theta'_i),\right.\\
& \hspace{1.2in}\left. R
\sin(\tfrac{2\pi(i-1)}{n}+2\theta'_i\sum_{j=0}^{k-1}2^{-j}
+2^{-k}\theta'_i),*\right).
\end{align*}

A sum of the form $\sum_{j=0}^{-1}  2^{-j}$ is taken to be equal to
$0$. Having defined the position points, we now give the preference
points. These are also defined using the $\rho$-cycles. For the
$\rho$-cycle $(i_1,\ldots,i_d)$, the preference points of
$\widehat{Q}_{i_1},\ldots,\widehat{Q}_{i_d}$ are $0$ in the first
two coordinates. Since every person has the position of his first
two coordinates on a circle of radius $R$ centred at the origin, the
distance between the preference point of a $Q_*$ person and the
position point of another person depends only on the  third
coordinate. In the third coordinate, for $1\leq j < d$,
$\widehat{Q}_{i_j}$ is placed between $\overline{P}_{i_j}$ and
$\overline{P}_{i_{j+1}}$, slightly closer to $\overline{P}_{i_j}$.
Here is the definition.
 For $f_i \in Rep(\rho)$ ,  for $0 \leq  k \leq q_i - 1$, set
 $$
\widehat{Q}_{\rho^k f_i} = \left(0,0, (i-1)/n+7k\theta_i + 3\theta_i\right).  $$

The preference points $\widehat{P}_{i_1},\ldots,\widehat{P}_{i_d}$
are $0$ in the third coordinate. We note that the contribution of
this coordinate towards the distance between its preference point
and some position point is at most~1. The preference points of the
$P$ people are placed around a circle of radius $R$ in the $x$-$y$
plane. We pick the radius $R$ large enough, say $10^{10^{4m}}$, to
ensure that the  third coordinate does not influence the ordering of
the initial part of the preference list of a $P$ person. In other
words, the only coordinates that affect the order are the $x$ and
the $y$ coordinates. Since the $x$ and $y$ coordinates of the
positions of the $P_*$ people, the $Q_*$ people and the preference
positions of the $P_*$ people all lie on the circle of radius $R$,
the position of a person $Y$ in the initial part of the preference
list of a person $P_i$ is purely a function of the angle subtended
by the preference point of $P_i$ and the position point of $Y$ at
the origin.

 If the $i$'th $\rho$-cycle is $(i_1,\ldots,i_d)$, then,
in the  first two coordinates, for $j>1$, $\widehat{P}_{i_j}$ is
placed $1/3$ of the way along between $\overline{Q}_{i_{j-1}}$ and
$\overline{P}_{\psi(i_j)}$. Then, $\widehat{P}_{i_1}$ is placed
between $\overline{Q}_{i_d}$ and $\overline{P}_{\psi(i_1)}$,
slightly nearer to $\overline{Q}_{i_d}$. Here is the definition,
which is similar to the definition for the last two coordinates in
Section~\ref{section:begin-construction}. For  $f_i \in Rep(\rho)$,
for  $0 \leq  k \leq q_i - 1$, set
\begin{align*}
\widehat{P}_{\rho^{k+1} f_i} =  \left( R \cos(
\tfrac{2\pi(i-1)}{n}+2\theta'_i\sum_{j=0}^{k-1}2^{-j}  + \tfrac13
2^{-k} \theta'_i), \right.\\
\hspace{1.2in}\left. R \sin(
\tfrac{2\pi(i-1)}{n}+2\theta'_i\sum_{j=0}^{k-1} 2^{-j} + \tfrac13
2^{-k} \theta'_i ),*\right).
\end{align*}

It is now easy to check that the preference lists are the same as those in   (\ref{eq:noisy-Q-pref-lists}) and (\ref{eq:noisy-P-pref-lists}). As in Section~\ref{section:begin-construction}, there is a slight issue
because $Q_i$ is indifferent between the other $Q_*$ people.
However, this can be fixed in the same way that it was fixed in the remark at the end of Section~\ref{sect:preference-vectors}.
The rest of the proof is now identical to Section~\ref{section:begin-construction}.

\end{proof}

\section{$\oneSRa$ is easy}

The goal of this section is to prove Theorem~\ref{thm:1SR}
which we restate below.

\setcounter{counter:save}{\value{theorem}}
\setcounter{theorem}{\value{counter:1SR}}
\begin{theorem}
For every $\oneSRa$ instance~$I$, there are either~$1$ or~$2$ stable
assignments. Thus, $\oneSRa$ can be solved exactly in polynomial time.
\end{theorem}
\setcounter{theorem}{\value{counter:save}}

\begin{proof}

We will show that for any $\oneSRa$ instance~$I$, there are either~$1$ or~$2$ stable assignments.
Finding these stable assignments can be done in polynomial time using Gusfield's algorithm~\cite{gusfield-structure}.
Assume that the instance has $n$~people and, without loss of
generality, assume that the positions of these people are ordered $1,\ldots,n$.
The preference of each person is either ``type A'', in which case his preference
list is $1,\ldots,n$, excluding himself, or his preference is ``type B'', in which case his preference list
is $n,\ldots,1$, excluding himself. The proof is by induction on~$n$. In the cases below, the notation
``$i$A'' means that, in instance~$I$, person~$i$ has a type-A list
(and $i$B is defined similarly). We start with some base cases.

\textbf{Base Cases.}

\begin{enumerate}[(B1)]
\item
\textbf{$n=2$}:\quad There is a single stable  assignment in which the
two people are paired.
\item
\textbf{$n=4$. $1$A, $2$B, $3$B, $4$A}\quad \label{ABBA} The lists
start out as
$$\begin{array}{c|ccc}
1 & 2 & 3 & 4\\
2 & 4 & 3 & 1\\
3 & \emph{4} & 2 & 1\\
4 & 1 & 2 & \emph{3}\\
\end{array}$$
Then~$2$ becomes semi-engaged to~$4$ so the lists are
$$\begin{array}{c|ccc}
1 & \emph{2} & {3} & 4\\
2 & 4 & 3 & \emph{1}\\
3 & 2 & 1\\
4 & 1 & 2\\
\end{array}$$
Then~$3$ becomes semi-engaged to~$2$ so the lists are
$$\begin{array}{c|ccc}
1 &  3 & 4\\
2 & 4 & 3 \\
3 & 2 & 1\\
4 & 1 & 2\\
\end{array}$$
Now everybody is semi-engaged, so phase~1 ends. We have two exposed
rotations,
$$R = \begin{array}{c|cc}
1 & 3 & 4\\
2 & 4 & 3\\
\end{array} \quad
R^d = \begin{array}{c|cc}
4 & 1 & 2\\
3 & 2 & 1\\
\end{array}$$
These lead to the two stable assignments $1-4,2-3$ and $1-3,2-4$.

 \item \textbf{$n= 4$. $1$B, $2$A, $3$A, $4$B}\quad
 This is symmetric to Case~(B\ref{ABBA}).
 The symmetry is as follows. Instead of ordering the positions in order $1,2,3,4$
 and then, in that order, assigning the lists consistent, inconsistent, inconsistent, consistent, as in Case~(B\ref{ABBA}),
 think about the backwards order $4,3,2,1$ and  assign the lists, in this order, as consistent, inconsistent,
 inconsistent, consistent.

\item \textbf{$n = 4$. $1$A, $2$B, $3$A, $4$B}\quad
The lists start out as
$$\begin{array}{c|ccc}
1 & 2 & 3 & \emph{4}\\
2 & 4 & 3 & 1\\
3 & {1} & 2 & 4\\
4 & 3 & 2 & \emph{1}\\
\end{array}$$
Then $3$ becomes semi-engaged to~$1$ removing pair $(1,4)$ so the
lists are
$$\begin{array}{c|ccc}
1 & 2 & 3 \\
2 & 4 & 3 & 1\\
3 & {1} & 2 & 4\\
4 & 3 & 2 \\
\end{array}$$
Now everybody is semi-engaged, so phase~1 ends. We have one exposed
rotation,
 $$R = \begin{array}{c|cc}
1 & 2 & 3\\
4 & 3 & 2\\
\end{array}  $$   leading to the unique stable assignment $1-3,2-4$.

\end{enumerate}

We now give the inductive argument.
The proof is easy, apart from checking that all cases are covered.
We start by enumerating some cases that
can occur, which we call ``top inductive step cases'' and ``mixed inductive step cases''.
Then we give some more cases, called ``symmetric inductive step cases''. These
are symmetric to cases that we've already done, so don't need a new argument.
Finally, we conclude with some accounting to check that all cases are covered.

\textbf{Top Inductive Step Cases} (Inductive cases that reduce the
number of people, by pairing off two people at the top of the lists)

\begin{enumerate}[(\textrm{T}1)]
\item
\textbf{$n\geq 4$. $1$A, $2$A}\quad \label{AA} Persons~$1$ and~$2$
prefer each other, so in phase~1, $1$ becomes semi-engaged to~$2$
and $2$ becomes semi-engaged to~$1$. Then the problem is reduced to
an instance of size $n-2$ without people~$1$ and~$2$.

\item
\textbf{$n\geq 4$.  $1$B, $2$A, $3$A, $i$A  for some $i>3$.}\quad
\label{BAnA} Person~$i$ becomes semi-engaged to person~$1$, removing
pairs $((1,i-1),\ldots,(1,2))$. Now persons~$2$ and~$3$ prefer each
other and can pair off. This gives us an intermediate instance,
$I'$, with two fewer people. $I'$ is not a  $1$-attribute instance. However,
the  $1$-attribute instance $I''$ derived from~$I$ by removing~$2$ and~$3$
has the same stable assignments as~$I'$ since it reaches~$I'$ by
having~$i$ semi-engaged to~$1$.

 \item \label{z}
 \textbf{$n\geq 6$. $1$A, $2$B, $3$A, $(n-1)$A, $n$B}
 Here $3$ becomes semi-engaged to~$1$
 removing $(1,4),\ldots,(1,n)$ (including especially
 $(1,n-1)$).
  Then $n-1$ becomes semi-engaged to~$2$ removing $(2,1),\ldots,(2,n-2)$
 (including especially
 $(2,1)$).
 Now~$1$ and~$3$ can pair off.
 This gives us an intermediate instance~$I'$, with $2$~fewer people.
 $I'$ is not a  $1$-attribute instance, so we have to show that the  $1$-attribute instance~$I''$
 derived from~$I$ by removing~$1$ and~$3$
also gets to~$I'$. This happens by making $n-1$ semi-engaged to~$2$
(which is now his first choice).

  \end{enumerate}

\textbf{Mixed Inductive Step Cases}
\begin{enumerate}[(\textrm{M}1)]
\item \label{opp}
\textbf{$n\geq 4$. $1$B, $n$A}\quad This case reduces to an instance
of size $n-2$ (without people~$1$ and~$n$) similar to
Case~(T\ref{AA}).

\item
\textbf{$n\geq 4$. $1$A, $2$B, $i$A, $n$A for some $2<i<n$.}\quad
\label{twon} From instance~$I$, person~$i$ becomes semi-engaged
to~$1$. This removes
  the pairs $(1,i+1),\ldots,(1,n)$  from the lists.
Now persons~$n$ and $2$ prefer each other, and pair off. This gives
an instance $I'$ of size $n-2$ (without people~$2$ and~$n$). $I'$ is
not quite a  $1$-attribute instance, because the initial semi-engagement
of~$i$ to~$1$ knocked out the pairs $(1,i+1),\ldots,(1,n-1)$ from
the lists. However, let $I''$ be the  $1$-attribute $(n-2)$-person instance
derived from instance~$I$ by deleting people~$2$ and~$n$. Note that,
from $I''$, person~$i$ becomes semi-engaged to~$1$ and then we are
at instance~$I'$. Thus, the stable  assignments of~$I$ are the stable
 assignments of~$I''$.

\item
\label{zzz} \textbf{$n\geq 6$. $1$A, $2$B, $3$B, $(n-2)$A, $(n-1)$A,
$n$B} Here~$3$ becomes semi-engaged to~$n$, removing $(n,2)$ and
$(n,1)$ so~$2$ now prefers~$n-1$. Then~$n-2$ becomes semi-engaged
to~$1$, removing $(1,n-1)$ and $(1,n)$ so $n-1$ now prefers~$2$.
Then~$2$ and~$n-1$ can pair off. Suppose that we started from the
original instance~$I$ and we just removed people~$2$ and~$n-1$ then
we can get to this ``paired off'' state by having~$3$ become
semi-engaged to~$n$ and having~$n-2$ become semi-engaged to~$1$.

\end{enumerate}

\textbf{Symmetric Inductive Step Cases}
\begin{enumerate}[(\textrm{S}1)]
\item \label{BB}
\textbf{$n\geq 4$. $(n-1)$B, $n$B}\quad Symmetric to~(T\ref{AA}).
People~$n-1$ and~$n$ pair off.

\item \label{stwon}
\textbf{$n\geq 4$. $1$B, $i$B, $(n-1)$A, $n$B for some
$1<i<n-1$.}\quad This is symmetric to Case~(M\ref{twon}). The
semi-engagement from~$i$ to~$n$ removes the pairs
$(i-1,n),\ldots,(1,n)$ so~$1$ and~$n-1$ pair off.

\item
\label{ok} \textbf{$n\geq 4$.  $i$B, $(n-2)$B, $(n-1)$B $n$A for
some $i<n-2$.}\quad Symmetric to case~(T\ref{BAnA}). Person~$i$
becomes semi-engaged to person~$n$, removing pairs
$((i+1,n),\ldots,(n-1,n))$. Now persons~$n-2$ and~$n-1$ prefer each
other and can pair off.

\item\label{zz}
\textbf{$n\geq 6$. $1$A, $2$B, $(n-2)$B, $(n-1)$A, $n$B} Symmetric
to Case~(T\ref{z}). Here~$n-2$ becomes semi-engaged to~$n$ removing
$(1,n),\ldots,(n-3,n)$ (including especially~$(2,n)$. Then $2$
becomes semi-engaged to~$n-1$ removing $(3,n-1),\ldots,(n,n-1)$
(including especially $(n,n-1)$). Now $n-2$ and $n$ can pair off.

\end{enumerate}

Now let's see that we've covered all cases for $n\geq 4$. Start by
looking at all~$16$ cases for lists~$1$, $2$, $n-1$ and~$n$.
Cases~(T\ref{AA}), (M\ref{opp}) and (S\ref{BB}) cover all
possibilities except
\begin{enumerate}[(\textrm{C}1)]
\item $1$A $2$B $(n-1)$A $n$A \label{doneM2}
\item $1$A $2$B $(n-1)$A $n$B\label{y1}
\item $1$A $2$B $(n-1)$B $n$A\label{y2}
\item $1$B $2$A $(n-1)$A $n$B\label{y3}
\item $1$B $2$B $(n-1)$A $n$B\label{other}
\end{enumerate}
Also, (C\ref{doneM2}) is covered by Case~(M\ref{twon}) and
(C\ref{other}) is covered by Case~(S\ref{stwon}). For $n=4$,
(C\ref{y1}),
(C\ref{y2})
and (C\ref{y3}) are explicitly covered by
our $n=4$ base cases. For $n=6$, they are covered as follows.

\begin{itemize}
\item  (C\ref{y2}):
  If any of $3,\ldots,n-2$ is an~A, then use Case~(M\ref{twon}).
Otherwise, use Case~(S\ref{ok}).
\item (C\ref{y3}):
Symmetrically,
 if any of $3,\ldots,n-2$ is a~B, then use Case~(S\ref{stwon}). Otherwise, use Case~(T\ref{BAnA})
 \item (C\ref{y1}):
  This is covered by cases (T\ref{z}) and (S\ref{zz}) and (M\ref{zzz})
\end{itemize}
\end{proof}

\section{$\BIS$ hardness of $\threeSRa$ and \\ $\twoSRe$}
\label{sec:BIShardness}

\subsection {$\threeSRa$}\label{sect:threeSRa}

The goal of this section is to prove Theorem~\ref{thm:BIS-3SR}
which we restate below.

\setcounter{counter:save}{\value{theorem}}
\setcounter{theorem}{\value{counter:BIS-3SR}}
\begin{theorem}
$\BIS \APred \threeSRa$.
\end{theorem}
\setcounter{theorem}{\value{counter:save}}

\begin{proof}

We re-use the construction that we used in~\cite{CGM} to reduce
$\BIS$ to counting stable assignments in the 3-attribute stable
marriage  model.

In \cite[Section 4.1.1]{CGM}, we  show how to take a $\BIS$ instance
$G=(V_1 \cup V_2,E)$, where $E \subseteq V_1 \times V_2$ and $|E|=n$
and turn it into a 3-attribute stable matching instance~$I^*$ with
$3n$ men and $3n$ women which are denoted
$\{A_1,\ldots,A_n,B_1,\ldots,B_n,\;\;$ $C_1,\ldots,C_n\}$ and
$\{a_1,\ldots,a_n,b_1,\ldots,b_n,c_1,\ldots,c_n\}$, respectively. We
associate two permutations, $\rho$ and $\sigma$ of $[n]$ with the
BIS instance. We show that independent sets of~$G$ are in one-to-one
correspondence with stable assignments of~$I^*$.

Using the permutations~$\rho$ and~$\sigma$ we will now show how to
modify the construction to obtain a $3$-attribute stable roommate
instance~$I$ with people
$$\{A_1,\ldots,A_n,B_1,\ldots,B_n,C_1,\ldots,C_n\} \cup
\{a_1,\ldots,a_n,b_1,\ldots,b_n,c_1,\ldots,c_n\}$$ so that stable
assignments for our instance~$I$ are in one-to-one correspondence
with stable assignments for~$I^*$.
Even though the stable roommate instance~$I$ simply has $6n$ people
(rather than having $3n$ men and $3n$ women, we will refer to the people
$\{A_1,\ldots,A_n,B_1,\ldots,B_n,C_1,\ldots,C_n\}$
as ``men'' and the people
$\{a_1,\ldots,a_n,b_1,\ldots,b_n,c_1,\ldots,c_n\}$
as ``women'' to simplify some of the descriptions.

Position and preference vectors for the instance~$I$ are defined
similarly to those of~$I^*$ from \cite[Section 4.1.2]{CGM}. We
reproduce here the required notation from \cite{CGM}: the $i$'th
$\sigma$-cycle has length~$p_i$, the $i$'th $\rho$-cycle has
length~$q_i$, $Rep(\rho)$ is the set consisting of one
representative for each cycle in $\rho$, and $Rep(\sigma)$ is the
set consisting of one representative for each cycle in $\sigma$.
The number of $\sigma$ cycles is $\ell$. The number of $\rho$ cycles is $k$.
The
only difference between the construction used here and the
construction used  in~\cite{CGM} is that, instead of using the full unit
circle in the first two coordinates to position people, we instead
use an angle of $\zeta < \pi/4$ to fit in the position vectors of
the
women
and the preference vectors of the
men.
The third coordinate of every  vector remains
unaltered.

\begin{align*}
\textrm{Let}\ \epsilon &= \frac{\zeta}{n^2}.\\
\textrm{For}\  e_i &\in Rep(\sigma), \; \textrm{let}\
\theta_i = \epsilon/(7p_i-1).\;\; \textrm{Then for}\ 0\leq m\leq p_i-1\ \textrm{define}\\
 \bar{a}_{\sigma^{m}{e_i}} &= \left(\cos(\zeta(i-1)/l + 7m\theta_i +
4\theta_i), \;\;\sin(\zeta(i-1)/l + 7m\theta_i +
4\theta_i), \;\;0\right), \\
\bar{b}_{\rho\sigma^{m}{e_i}} &= (\cos(\zeta(i-1)/l + 7m\theta_i +
6\theta_i), \;\;\sin(\zeta(i-1)/l + 7m\theta_i + 6\theta_i), \;\;
4^{\rho\sigma^{m}{e_i}}),\;\\
\textrm{and\;\;\;} & \\
\bar{c}_{\sigma^{m-1}{e_i}} &= \left(\cos(\zeta(i-1)/l +
7m\theta_i), \;\;\sin(\zeta(i-1)/l + 7m\theta_i), \;\;0\right).
\end{align*}

\begin{align*}
\textrm{Let}\ \phi &= 2\pi /100\ \textrm{and}\ \epsilon = \frac{\zeta}{n^2}.\;\; \\
\textrm{For}\ e_i &\in Rep(\sigma),\ \textrm{let}\ \theta_i = \epsilon/(7p_i-1).\;\; \textrm{Then for}\ 0\leq m\leq p_i-1\ \textrm{define}\\
 \hat{A}_{\sigma^{m}{e_i}} &= \left(\cos(\zeta(i-1)/l + 7m\theta_i +
(14/3)\theta_i),\right.\\
&\hspace{1.2in}\left.\sin(\zeta(i-1)/l + 7m\theta_i +
(14/3)\theta_i),\, 0\right), \\
\hat{B}_{\sigma^{m}{e_i}} &= \left(\sin \phi \cos(\zeta(i-1)/l +
7m\theta_i + 4\theta_i), \right.\\
&\hspace{1.2in}\left.\;\;\sin \phi \sin(\zeta(i-1)/l + 7m\theta_i +
4\theta_i),
\;\;\cos \phi \right),\;\\
\textrm{and}\;\;\; &\\
\hat{C}_{\sigma^{m-1}{e_i}} &= \left(\cos(\zeta(i-1)/l + 7m\theta_i
+ (8/5)\theta_i), \;\sin(\zeta(i-1)/l + 7m\theta_i + (8/5)\theta_i),
\;0\right).
\end{align*}

We note that $\epsilon$ and the $2\pi$ factor in the cosine and sine
terms have been scaled appropriately (relative to the construction in~\cite{CGM}) by the factor $(\zeta/2\pi)$.
If we were to restrict the preference lists of
the men
to the set of
women,
then
these preference lists would match the lists of the
men from the
stable marriage instance $I^*$ of~\cite{CGM}.

Since we are working with a stable
roommate instance, we also have to place the position vectors of the
the men
and the preference vectors of the
women
in the same 3-dimensional space. We take the position
vectors of the
men
and the preference vectors of
the
women
from the stable marriage instance~$I^*$ and
modify them as follows:
\begin{enumerate}[(i)]
\item We scale down the $\epsilon$ and $2\pi$
terms in the sine and cosine terms by a factor of $(\zeta/2\pi)$,
\item we offset the angle in the sine and cosine terms by an angle of
$\pi$, and
\item we negate the  third coordinate of all the  position and preference vectors.
\end{enumerate}  We summarise the position vectors of the
men
and the preference vectors of the
women
in the stable roommate instance~$I$
below.

\begin{align*}
\textrm{Let}\ \epsilon &= \frac{\zeta}{n^2}. \\
\textrm{For}\ f_i &\in Rep(\rho),\; \textrm{let}\ \omega_i = \epsilon/(7q_i-1).\; \textrm{Then for}\ 0\leq m\leq q_i-1\;\textrm{we define}\\
\bar{A}_{\rho^{m-1}{f_i}} &= \left(\cos(\pi+ \zeta(i-1)/k +
7m\omega_i),
\;\;\sin(\pi+ \zeta(i-1)/k + 7m\omega_i), \;\;0\right), \\
\bar{B}_{\rho^{m}{f_i}} &= \left(\cos(\pi+ \zeta(i-1)/k + 7m\omega_i
+4\omega_i),\right.\\
&\hspace{1.2in}\left.\;\sin(\pi+ \zeta(i-1)/k + 7m\omega_i +
4\omega_i), \;0\right), \;\\
\textrm{and} \;\;\;&\\
\bar{C}_{\rho^{m}{f_i}} &= \left(\cos(\pi+ \zeta(i-1)/k + 7m\omega_i
+ 6\omega_i), \right.\\
&\hspace{1.2in}\left.\sin(\pi+ \zeta(i-1)/k + 7m\omega_i +
6\omega_i), \;\;-4^{\rho^{m}{f_i}}\right).
\end{align*}

\begin{align*}
\textrm{Let}\ \phi &= 2\pi /100\ \textrm{and}\ \epsilon = \frac{\zeta}{n^2}.\\
\textrm{For}\ f_i &\in Rep(\rho), \;\textrm{let}\ \omega_i = \epsilon/(7q_i-1).\; \textrm{Then for}\ 0\leq m\leq q_i-1\; \textrm{we define}\\
\hat{a}_{\rho^{m}{f_i}} &= (\sin \phi \cos(\pi+ \zeta(i-1)/k +
7m\omega_i + 4\omega_i), \;\;\\
&\;\;\;\;\;\hspace{0.4in}\sin \phi \sin(\pi+ \zeta(i-1)/k + 7m\omega_i + 4\omega_i), \;\;-\cos \phi),\;  \\
\hat{b}_{\rho^{m}{f_i}} &= \left(\cos(\pi+ \zeta(i-1)/k + 7m\omega_i
+ (8/5)\omega_i)\right.,
\;\;\\
&\;\;\;\;\;\hspace{0.4in} \left.\sin(\pi+ \zeta(i-1)/k + 7m\omega_i + 8/5\omega_i), \;\;0\right),\; \textrm{and}\\
\hat{c}_{\rho^{m}{f_i}} &= \left(\cos(\pi+ \zeta(i-1)/k + 7m\omega_i
+
14/3\omega_i), \;\;\right.\\
&\;\;\;\;\;\hspace{0.4in}\left.\sin(\pi+ \zeta(i-1)/k + 7m\omega_i + (14/3)\omega_i), \;\;0\right). \\
\end{align*}

As observed before, if we were to restrict the preference lists of
the women
to the set of
men
then these preference lists would match the lists of the
women
from the stable marriage instance $I^*$
from \cite{CGM}.

We now state a very simple lemma which  we will use to  connect our stable
roommate  instance
$I$
to the stable matching instance~$I^*$.

\begin{lemma}
Suppose vector $\bar{v}_1 = (cos \theta_1, sin \theta_1, \alpha_1)$
and vector $\bar{v}_2 = (cos \theta_2, sin \theta_2,$ $\alpha_2)$.

(i)If $\pi/2< |\theta_1 - \theta_2| < 3\pi/2$ and $\alpha_1\cdot
\alpha_2 \leq 0$, then $\bar{v}_1 \cdot \bar{v}_2 < 0$.

(ii)If $0< |\theta_1 - \theta_2| < \pi/4$ and $\alpha_1\cdot
\alpha_2 \geq 0$, then $\bar{v}_1 \cdot \bar{v}_2 > 0$.
\end{lemma}
\begin{proof}
(i)The dot product $\bar{v}_1 \cdot \bar{v}_2 = cos (\theta_1 -
\theta_2) + \alpha_1\cdot \alpha_2 < 0$.

(ii)The dot product $\bar{v}_1 \cdot \bar{v}_2 = cos (\theta_1 -
\theta_2) + \alpha_1\cdot \alpha_2 > 0$.
\end{proof}

Applying this lemma to the preference vectors and the position
vectors of the
men,
we see that the dot
product of a preference vector of
a man
with the position vector of
a man
is always negative. Similarly, the dot product of a preference vector
of
a man
with the position vector of
a woman
is always positive. This implies that
the initial $n$ positions on the preference lists of
the men
would be populated by
the women
and would coincide with the preference lists of the
the men
in the stable matching instance $I^*$ from~\cite{CGM}. The
same holds true for the preference lists of
the women.
The first $n$ positions on their preference lists are
occupied by
the men
and this initial part of
their preference lists matches the preference lists of the
women
from the stable matching instance $I^*$.

To finish the proof we show that the stable matchings of~$I$ are in one-to-one correspondence
with the stable matchings of~$I^*$.

First, suppose that~$M$ is a stable matching of~$I^*$. It is clear that~$M$ is a matching of~$I$,
so we must show that it is stable for~$I$.
Since $M$ is stable for~$I^*$,  it has no man-woman blocking pairs.
Also, it is easy to see that, in~$I$, there is no man-man blocking pair
(since each man prefers all women to the other men) and similarly, there is no woman-woman blocking pair. Thus, $M$ is a stable matching of~$I$.

Next, suppose that~$M$ is  a stable matching of~$I$.  First, we show that~$M$ is
a valid matching of~$I^*$ --- that is, every matched pair consists of one man and one woman.
Suppose instead that two men~$P_i$ and~$P_j$ are matched in~$I$. By the pigeonhole
principle, two women $p_k$ and $p_\ell$ must also be matched in~$I$ but now $(P_i,p_\ell)$ form
a blocking pair. Thus, $M$ is a matching of~$I^*$. Since $M$ has no blocking pairs in~$I$
it also has no blocking pairs in~$I^*$, so it is a stable matching of~$I^*$.

Thus, the set of stable assignments for the roommate instance~$I$ is
identical to the set of stable assignments (matchings) for the
stable marriage instance~$I^*$. We have already shown in~\cite{CGM}
that the latter is in one-to-one correspondence with the independent
sets of~$G$, completing the proof.

\end{proof}

\subsection {$\twoSRe$}
\label{sec:blahblah}

The goal of this section is to prove Theorem~\ref{thm:BIS-2SRe}
which we restate below.

\setcounter{counter:save}{\value{theorem}}
\setcounter{theorem}{\value{counter:BIS-2SRe}}
\begin{theorem}
$\BIS \APred \twoSRe$.
\end{theorem}
\setcounter{theorem}{\value{counter:save}}

\begin{proof}

As in section \ref{sect:threeSRa}, we  re-use  the construction
that we used in~\cite{CGM}
to reduce $\BIS$ to counting stable assignments in the $2$-Euclidian stable marriage model.

In \cite[Section  6]{CGM}, we  show how to take a $\BIS$ instance
$G=(V_1 \cup V_2,E)$, where $E \subseteq V_1 \times V_2$ and $|E|=n$
and turn it into a  $2$-Euclidian stable matching instance~$I^*$
with $3n$ men and $3n$ women which are denoted
$\{A_1,\ldots,A_n,B_1,\ldots,B_n,C_1,\ldots,C_n\}$ and
$\{a_1,\ldots,a_n,b_1,\ldots,b_n,c_1,\ldots,c_n\}$, respectively, as
before. Once again, we associate two permutations, $\rho$ and
$\sigma$ of $[n]$ with the BIS instance. We show that independent
sets of~$G$ are in one-to-one correspondence with stable assignments
of~$I^*$.

Using the permutations~$\rho$ and~$\sigma$ we will now show how to
modify the construction to obtain a  $2$-Euclidian stable roommate
instance~$I$ with people
$$\{A_1,\ldots,A_n,B_1,\ldots,B_n,C_1,\ldots,C_n\} \cup
\{a_1,\ldots,a_n,b_1,\ldots,b_n,c_1,\ldots,c_n\}$$ so that stable
assignments for our instance~$I$ are in one-to-one correspondence
with stable assignments for~$I^*$. We will refer to the people
$\{A_1,\ldots, A_n,B_1, \ldots,B_n,\;\;$ $C_1, \ldots,C_n\}$ as
``men'' in~$I$ and to the people
$\{a_1,\ldots,a_n,b_1,\ldots,b_n,c_1,\ldots,c_n\}$ as ``women'' in
order to make it easier to describe the construction.

In the stable roommate construction, the preference points of the
men and the position points of the  women are the same as those in
the 2-Euclidean stable marriage construction from~\cite{CGM}. The
preference points of the women are obtained from those in the stable
marriage instance by negating both coordinates. Also, the position
points of the men are obtained from those in the stable marriage
instance by negating both coordinates. Before we define the
positions, we wish to remind the reader that the $i$'th
$\sigma$-cycle (out of $\ell$ $\sigma$ cycles) has length~$p_i$, the
$i$'th $\rho$-cycle (out of $k$ $\rho$ cycles) has length~$q_i$,
$Rep(\rho)$ is the set consisting of one representative for each
cycle in $\rho$, and $Rep(\sigma)$ is the set consisting of one
representative for each cycle in $\sigma$ as in \cite{CGM}.

The positions are defined as follows. For $ e_j \in Rep(\sigma)$, $
f_j \in Rep(\rho)$, $0 \leq h \leq p_j-1$, and $0 \leq g \leq q_j-1$
we let
\begin{align*}
\bar{a}_{\sigma^{h} e_j} &= \left(\;\sum_{i=0}^{j-1}2p_i + h+1\;,\;0\right), \\
\bar{b}_{\rho\sigma^{h} e_j} &= \left(0\;,\;\sum_{i=0}^{j-1}2p_i + h+1\right),\\
\bar{c}_{\sigma^{(h-1)} e_j} &= \left(\;\sum_{i=0}^{j-1}2p_i +
h+0.3\;,\;0\right),\\
\bar{A}_{\rho^{g-1} f_j} &= \left(\;-\sum_{i=0}^{j-1}2q_i - g-0.3\;,\;0\right),\\
\bar{B}_{\rho^{g} f_j} &= \left(\;-\sum_{i=0}^{j-1}2q_i - g-1\;,\;0\right), \mbox{ and}\\
\bar{C}_{\rho^{g} f_j} &= \left(0\;,\;-\sum_{i=0}^{j-1}2p_i -g-1\right).
\end{align*}

The preferences are defined as follows.
Let $\epsilon =  1/100^{n}$. For $e_j  \in Rep(\sigma)$,
$f_j  \in Rep(\rho)$, $0 \leq h \leq p_j-1$, $0 \leq g \leq q_j-1$, we let

\begin{align*}
\hat{A}_{\sigma^{h} e_j} &= \left(\;\sum_{i=0}^{j-1}2p_i + h+1\;,\;\sum_{i=0}^{j-1}2p_i + h+1 - \epsilon\right), \\
\hat{B}_{\sigma^{h} e_j} &= \left(\;\sum_{i=0}^{j-1}2p_i + h+1\;,\;1000^{n}\right),   \\
\hat{C}_{\sigma^{(h-1)} e_j} &= \left(\;\sum_{i=0}^{j-1}2p_i +
h+0.6\;,\;0\right),\\
\hat{a}_{\rho^{g} f_j} &= \left(\;-\sum_{i=0}^{j-1}2q_i - g-1\;,\;\;-1000^{n}\right), \\
\hat{b}_{\rho^{g} f_j} &= \left(\;-\sum_{i=0}^{j-1}2q_i - g-0.6\;,\;0\right),  \mbox{ and} \\
\hat{c}_{\rho^{g} f_j} &= \left(\;-\sum_{i=0}^{j-1}2q_i -
g-1\;,\;-\sum_{i=0}^{j-1}2q_i - g-1 + \epsilon\right).
\end{align*}

We show in \cite[Section 6]{CGM} that the preference lists of the stable matching instance~$I^*$
have prefixes as described as follows, where $\tau$ is a permutation of $[n]$ (we won't
need the details of~$\tau$ in this paper).
In all stable matchings, the men and women of~$I^*$
are matched to partners which are included in these prefixes.
We refer to these lists as the ``initial'' preference lists of~$I^*$.

\begin{align}
\label{women}
\textrm{for}\ f_i&\in Rep(\rho),\\
\nonumber
b_{\rho^m f_i} \;\;&:\;\;\; A_{\rho^{(m-1)} f_i} B_{\rho^m
  f_i}\;\;,\;\;\;0\leq m \leq q_i-1, \\
\nonumber
c_{\rho^{m} f_i} \;\;&:\;\;\; B_{\rho^{m} f_i} C_{\rho^{m}
  f_i}\;\;,\;\;\;0\leq m \leq q_i-1, \\
\nonumber
a_{\rho^m f_i} \;\;&:\;\; C_{n} C_{n-1} \cdots C_{1}
B_{\rho^m f_i} A_{\rho^m f_i}
\;\;,\;\;\;0\leq m \leq q_i-2,\; \textrm{and}\\
\nonumber
a_{\rho^{(q_i-1)} f_i} \;\;&:\;\;
C_{n} C_{n-1} \cdots C_{1}
\\
\nonumber
&\;\;
B_{\rho^{(q_i-1)} f_i} A_{\rho^{(q_i-2)} f_i} B_{\rho^{(q_i-2)}
f_i}\cdots B_{\rho^{2} f_i}A_{\rho f_i} B_{\rho f_i}
A_{f_i} B_{f_i}
A_{\rho^{(q_i-1)}f_i}.
\end{align}
\begin{align}
\label{men}
\textrm{For}\ e_i&\in Rep(\sigma),
\\
\nonumber
A_{\sigma^m e_i} \;\;&:\;\;\; a_{\sigma^m e_i} b_{\rho \sigma^m
  e_i}\;\;,\;\;\;0\leq m \leq p_i-1, \\
\nonumber
C_{\sigma^{(m-1)} e_i} \;\;&:\;\;\; c_{\sigma^{(m-1)} e_i}
a_{\sigma^{m} e_i}\;\;,\;\;\;0\leq m \leq p_i-1, \\
\nonumber
B_{\sigma^m e_i} \;\;&:\;\; b_{\tau(n)} b_{\tau(n-1)} \cdots b_{\tau(1)}
a_{\sigma^m
  e_i} c_{\sigma^m e_i}\;\;,\;\;\;0\leq m \leq p_i-2,\; \textrm{and}
  \\
\nonumber
B_{\sigma^{(p_i-1)} e_i} \;\;&:\;\; b_{\tau(n)} b_{\tau(n-1)} \cdots
b_{\tau(1)}
a_{\sigma^{(p_i-1)} e_i} c_{\sigma^{(p_i-2)} e_i}
a_{\sigma^{(p_i-2)} e_i}\cdots a_{\sigma e_i}c_{e_i}
a_{e_i}c_{\sigma^{(p_i-1)} e_i}.
\end{align}

We now make three observations. The first of these is self-evident from the
construction.  We provide justifications
below
for Observations~2 and~3.
\begin{enumerate}

\item For the stable roommate instance~$I$, the preference lists of the  men, when restricted
to  women,   match the preference lists from the
stable marriage instance~$I^*$.  Similarly, the preference lists
of the women, when restricted to men, match
the preference lists from~$I^*$.

\item  For the stable marriage instance $I^*$, the distance between the preference position
of any man  and the   position point of any  woman on   his  initial
preference list is less than the distance between his  preference
position and the origin. Similarly, the distance between the
preference position of any woman  and the   position point of any
man on  her  initial preference list is less than the distance
between
 her  preference position and the origin.

\item For the stable roommate instance~$I$,  the distance between the
preference position of any man  and the origin is less than the
distance between his  preference position and the position point of
any other man. Similarly,  the distance between the preference
position of any woman  and the origin is less than the distance
between her  preference position and the position point of any other
woman.

\end{enumerate}

We now provide arguments that validate Observations~2 and~3. The
following calculations use the prefixes of the preference lists
of~$I^*$ from (\ref{women}) and (\ref{men}).
To establish Observation 2,  we show
that the preference point of a man is closer to the last woman on
his initial preference list than to the origin. Similarly, we show
that the preference point of a woman is closer to the last man on
her initial preference list than to the origin.

For man $A_{\sigma^m e_i}$, where $e_i \in Rep(\sigma)$ and $0 \leq
m \leq p_i-1$, the distance to woman $b_{\rho \sigma^m e_i}$ and the
origin ($\bar{0}$)are as follows.

\begin{eqnarray*}
d^2(\hat{A}_{\sigma^{m} e_i},\bar{b}_{\rho\sigma^{m} e_i}) & = &
(\sum_{j=0}^{i-1}2p_j + m+1- 0)^2 \\
& & + \ (\sum_{j=0}^{i-1}2p_j + m+1-\epsilon-\sum_{j=0}^{i-1}2p_j - m-1)^2\\
& = & (\sum_{j=0}^{i-1}2p_j + m+ 1)^2 + \epsilon^2
\end{eqnarray*}
\begin{eqnarray*}
d^2(\hat{A}_{\sigma^{m} e_i},\bar{0}) & = & (\sum_{j=0}^{i-1}2p_j +
m+1-
0)^2 \\
& & + \ (\sum_{j=0}^{i-1}2p_j + m+1-\epsilon-0)^2\\
& \geq & (\sum_{j=0}^{i-1}2p_j + m+1)^2 +  (1-\epsilon)^2 \\
& > & (\sum_{j=0}^{i-1}2p_j + m+ 1)^2 + \epsilon^2 =
d^2(\hat{A}_{\sigma^{m} e_i},\bar{b}_{\rho\sigma^{m} e_i})
\end{eqnarray*}

For man $C_{\sigma^{m-1} e_i}$, where $e_i \in Rep(\sigma)$ and $0
\leq m \leq p_i-1$, the distance to woman $a_{\sigma^m e_i}$ and the
origin are as follows.

\begin{eqnarray*}
d^2(\hat{C}_{\sigma^{(m-1)} e_i},\bar{a}_{\sigma^{m} e_i}) & = &
(\sum_{j=0}^{i-1}2p_j + m+0.6 - \ \sum_{j=0}^{i-1}2p_j - m-1)^2 + (0-0)^2 = 0.16.\\
d^2(\hat{C}_{\sigma^{(m-1)} e_i},\bar{0}) &= & (\sum_{j=0}^{i-1}2p_j
+ m+0.6 - 0)^2  + \ (0-0)^2 \geq 0.6^2 = 0.36\\
& > &  0.16 = d^2(\hat{C}_{\sigma^{(m-1)} e_i},\bar{a}_{\sigma^{m}
e_i})
\end{eqnarray*}

For man $B_{\sigma^m e_i}$, where $e_i \in Rep(\sigma)$, we consider
two cases: (i) $m \neq p_i-1$,  and (ii) $m = p_i-1$.

Case(i) $m \neq p_i-1$: From the preference position of $B_{\sigma^m
e_i}$, we compute distances  to $c_{\sigma^{m} e_j}$ and to the
origin.
\begin{eqnarray*}
d^2(\hat{B}_{\sigma^{m} e_i},\bar{c}_{\sigma^{m} e_{i}}) & = &
(\sum_{j=0}^{i-1}2p_j + m+1
-\sum_{j=0}^{i-1}2p_j - (m+1)-0.3)^2 +
(1000^n)^2\\
& = & 0.09 + 1000^{2n} \ \ \textrm { and }\\
d^2(\hat{B}_{\sigma^{m} e_i},\bar{0}) & = & (\sum_{j=0}^{i-1}2p_j +
m+1 -0)^2 + (1000^n)^2\\
& \geq & 1 + 1000^{2n} > 0.09 + 1000^{2n} = d^2(\hat{B}_{\sigma^{m}
e_i},\bar{c}_{\sigma^{m} e_{i}}).
\end{eqnarray*}

Case(ii) $m = p_i-1$: From the preference position of $B_{\sigma^m
e_i}$, we compute distances  to $c_{\sigma^{p_i-1} e_j} =
c_{\sigma^{0-1} e_j}$ and to the origin.
\begin{eqnarray*}
d^2(\hat{B}_{\sigma^{m} e_i},\bar{c}_{\sigma^{m} e_{i}}) & = &
(\sum_{j=0}^{i-1}2p_j + m+1 -\sum_{j=0}^{i-1}2p_j - 0.3)^2 +
(1000^n)^2\\
& = & (m + 0.7)^2 + 1000^{2n} = (p_i - 1 + 0.7)^2 + 1000^{2n}\\
& = &(p_i - 0.3)^2 + 1000^{2n}
\ \ \textrm { and }\\
d^2(\hat{B}_{\sigma^{m} e_i},\bar{0}) & = & (\sum_{j=0}^{i-1}2p_j +
m+1 -0)^2 + (1000^n)^2\\
& \geq & (m+1)^2 + 1000^{2n} = (p_i-1+1)^2 + 1000^{2n}\\
&=& (p_i)^2 + 1000^{2n} > (p_i-0.3)^2 + 1000^{2n} =
d^2(\hat{B}_{\sigma^{m} e_i},\bar{c}_{\sigma^{m} e_{i}}).
\end{eqnarray*}

The above set of computations and comparisons establish that the
preference point of a man is closer to the last   woman on his
initial preference list than to the origin. We can establish a
similar result for the women by repeating the above computations for
the preference position of every woman   and the position point of
the last
 man   on her initial preference list. Hence, we can conclude that
Observation~2 holds.

Next we establish Observation~3. We start by comparing the distance
between the preference point of any man   and the origin with the
distance between this preference point and the position point of an
$A_{*}$ or $B_{*}$ man. We note that the $x$-coordinate of the
position point of an $A_{*}$ or $B_{*}$ man is at most -0.3. We also
note that the $x$-coordinate of the preference point of  any man is
non-negative. In the equations to follow, $\hat{X}$ stands for the
preference point of a man and $\bar{Y}$ stands for the position
point of an $A_{*}$ or $B_{*}$ man. The $x$ and $y$ co-ordinates of
$\hat{X}$  will be denoted $\hat{X}_{x}$ and $\hat{X}_{y}$
respectively. The $x$ co-ordinate of $\bar{Y}$ will be denoted
$\bar{Y}_{x}$. As noted above, $\hat{X}_{x} \geq 0$ and $\bar{Y}_{x}
\leq -0.3$.

\begin{eqnarray*}
d^2(\hat{X} ,\bar{0}) & = & (\hat{X}_x -0)^2 + \ (\hat{X}_{y}-0)^2\\
& = & (\hat{X}_{x})^2 + \ (\hat{X}_{y})^2
\end{eqnarray*}
\begin{eqnarray*}
d^2(\hat{X} ,\bar{Y} ) & = & (\hat{X}_{x} - \bar{Y}_{x})^2 +
\
(\hat{X}_{y}-0)^2\\
& \geq &
(\hat{X}_{x} - (-0.3))^2 + \ (\hat{X}_{y}-0)^2\\
& = & (\hat{X}_{x} +0.3)^2 + \ (\hat{X}_{y})^2 \\
& > & (\hat{X}_{x})^2 + \ (\hat{X}_{y})^2 =
d^2(\hat{X} ,\bar{0})
\end{eqnarray*}

Next we compare the distance  between the preference point of a man
and the origin with the distance between this preference point and
the position point of a $C_*$ man. We note that the $y$-coordinate
of the position point of a $C_{*}$ man is at most -1. We also note
that the $y$-coordinate of the preference point of any
  man is non-negative. In the equations to
follow, $\hat{X} $ stands for the preference point of a   man. The
$x$ and $y$ co-ordinates of $\hat{X} $
   will be denoted $\hat{X}_{x}$ and $\hat{X}_{y}$ respectively.

\begin{eqnarray*}
 d^2(\hat{X} ,\bar{C}_{*}) & \geq &
(\hat{X}_{x} - 0)^2 + \ (\hat{X}_{y}-(-1))^2\\
& = & (\hat{X}_{x})^2 + \ (\hat{X}_{y} + 1)^2
\end{eqnarray*}
\begin{eqnarray*}
d^2(\hat{X} ,\bar{0}) & = & (\hat{X}_{x}-0)^2 + \ (\hat{X}_{y}-0)^2\\
& = & (\hat{X}_{x})^2 + \ (\hat{X}_{y})^2\\
& < & (\hat{X}_{x})^2 + \ (\hat{X}_{y} + 1)^2 \leq
d^2(\hat{X} ,\bar{C}_{*})
\end{eqnarray*}

The above   calculations establish that the preference point of any
man is closer to the origin than the position point of  any man. We
can establish a similar result for the women.
  Hence, we can conclude that
Observation~3 holds.

Combining Observations 1, 2 and 3, we note that
the prefixes of the preference lists of the stable roommate instance~$I$
are the same as those of the stable matching instance~$I^*$ from
(\ref{women}) and (\ref{men}).
We conclude that the set of stable assignments for the
roommate instance~$I$ is identical to the set of stable assignments
for the stable marriage instance~$I^*$. From~\cite{CGM}, we have
that the latter is in one-to-one correspondence with the independent
sets of~$G$, thereby, establishing the required result.

\end{proof}

\section*{Acknowledgements}

The authors wish to thank David Manlove for suggesting the problem
and for his useful discussions.

\newpage

\section*{Appendix 1 --- An Example}

We give the following example to illustrate the definitions from Section~\ref{sect:background}
and the two-phase algorithm for finding a stable roommate assignment.
Consider the following preference lists.

\begin{center}
  \begin{tabular}{r|rrrrrrrrrrr}
      1  & 12 &  7 &  4 &  6 &  9 &  5 & 10 &  2 &  3 &  8 & 11 \\
      2  &  5 &  6 &  1 &  9 & 12 &  4 &  3 & 10 &  8 & 11 &  7 \\
      3  & 11 &  9 &  4 &  1 &  8 & 12 &  2 &  6 &  5 &  7 & 10 \\
      4  &  2 &  9 & 12 & 10 &  7 &  6 &  1 &  8 &  5 & 11 &  3 \\
      5  & 12 &  6 &  3 &  9 &  4 &  10 & 11 &  8 &  7 &  2 &  1 \\
      6  &  8 &  4 &  1 & 10 &  2 & 11 &  3 &  5 & 12 &  7 &  9 \\
      7  &  3 &  5 &  2 &  6 & 10 &  4 & 11 &  1 &  8 &  9 & 12 \\
      8  &  1 &  7 & 10 & 12 &  3 &  2 &  5 &  4 &  9 &  6 & 11 \\
      9  &  2 & 12 &  1 &  6 &  5 & 11 &  8 & 10 &  3 &  7 &  4 \\
      10 &  1 &  4 &  3 & 11 &  2 &  7 &  6 &  8 &  9 &  5 & 12 \\
      11 &  6 &  4 &  8 & 10 & 12 &  5 &  3 &  1 &  2 &  7 &  9 \\
      12 & 11 &  6 &  3 &  2 &  7 &  4 &  9 & 10 &  1 &  5 & 8
  \end{tabular}
\end{center}

Phase I proceeds as outlined in Section~\ref{sect:poset}, with proposals
occurring and
``semi-engagements" forming.
Here are the short lists at the end of Phase I.
\begin{center}
  \begin{tabular}{r|rrrrrrrrrrr}
      1  & 7  & 4  & 6  & 9  & 10 &  &   &    \\
      2  & 6  & 9  &    &    &    &  &   &    \\
      3  & 9  & 8  & 12 & 5  &    &  &   &    \\
      4  & 12 & 10 & 7  & 6  & 1 & 8 & 5 & 11 & & & \\
      5  & 3  & 9  & 4  & 8  & 7 \\
      6  & 8  & 4  & 1  & 10 & 2 \\
      7  & 5  & 10 & 4  & 1 \\
      8  & 10 & 3  & 5  & 4  & 9 & 6 \\
      9  & 2  & 1  & 5  & 8  & 3 \\
      10 & 1  & 4  & 11 & 7  & 6 & 8 \\
      11 & 4  & 10 & 12 \\
      12 & 11 & 3  & 4
  \end{tabular}
  \begin{tabular}{cl}
      $R_1$ &
      \begin{tabular}{r|rr}
      3  &  9 & 8 \\
      6  &  8 & 4 \\
      11 &  4 & 10 \\
      8  & 10 & 3 \\
      5  & 3  & 9 \\
      \end{tabular}
  \end{tabular}
\end{center}
At the end of Phase I,
rotation~$R_1$ is exposed in the table (and no other rotations are exposed).
Note that $R_1$ is a singleton rotation.
After eliminating this rotation, we have the following table:
\begin{center}
  \begin{tabular}{r|rrrrrrrr}
      1  & 7  & 6  & 9  & 10 &  &   &    \\
      2  & 6  & 9  &    &    &    &  &   &    \\
      3  & 8  &  &   &    \\
      4  & 12 & 10 & 7  & 6  \\
      5  & 9  & 7 \\
      6  & 4  & 1  & 2 \\
      7  & 5  & 4  & 1 \\
      8  & 3 \\
      9  & 2  & 1  & 5 \\
      10 & 1  & 4  & 11 \\
      11 & 10 & 12 \\
      12 & 11 & 4 \\
  \end{tabular}
  \begin{tabular}{cl}
      $R_2$ &
      \begin{tabular}{r|rr}
         4  &  12 & 10 \\
         11 & 10 & 12 \\
      \end{tabular}  \\
      \\
      \\
      $R_3$ &
      \begin{tabular}{r|rr}
        1 & 7 & 6 \\
        2 & 6 & 9 \\
        5 & 9 & 7 \\
      \end{tabular}  \\
      \\
      \\
      $R_4$ &
      \begin{tabular}{r|rr}
         6 & 4 & 1 \\
        10 & 1 & 4 \\
      \end{tabular}  \\
  \end{tabular}
\end{center}

 There are   three rotations, $R_2, R_3$, and $R_4$,
exposed in this new table.  Using Definition~\ref{def:explicitly-precedes},
we find that,  in the rotation poset
for this instance, $R_1$ precedes each of $R_2, R_3$, and $R_4$.
For example, to see that~$R_1$ explicitly precedes~$R_2$, take $e_i=4$ and ad $p=2$.
Also, rotations $R_2$, $R_3$ and $R_4$ do not precede each other.
Each can be performed  from
the above table.

Also, each of $R_2, R_3$, and $R_4$ has a dual rotation.
For example, performing rotation~$R_4$ results in
the following table,
in which both $R_2^d$ and
$R_3^d$ are now exposed
(as are $R_2$ and $R_3$).
\begin{center}
  \begin{tabular}{r|rrrrrrrr}
      1  & 7  & 6  \\
      2  & 6  & 9  \\
      3  & 8   \\
      4  & 12 & 10  \\
      5  & 9  & 7 \\
      6  & 1  & 2 \\
      7  & 5  & 1 & & & & & \\
      8  & 3 \\
      9  & 2  & 5 \\
      10 & 4  & 11 \\
      11 & 10 & 12 \\
      12 & 11 & 4 \\
  \end{tabular}
  \begin{tabular}{cl}
      $R_2^d$ &
      \begin{tabular}{r|rr}
         10 &  4 & 11 \\
         12 & 11 & 4 \\
      \end{tabular}  \\
      \\
      \\
      $R_3^d$ &
      \begin{tabular}{r|rr}
        6 & 1 & 2 \\
        9 & 2 & 5 \\
        7 & 5 & 1 \\
      \end{tabular}
      \end{tabular}  \\
\end{center}

So in the rotation poset we have the relations
$\Pi^*(R_4, R_2^d)$ and $\Pi^*(R_4, R_3^d)$.
Recalling Theorem~\ref{lem:rotation-structure}, this also means that
$\Pi^*(R_2, R_4^d)$ and $\Pi^*(R_3, R_4^d)$.

By a careful analysis, we can determine that there are five stable roommate
assignments,
which we list next
along with
the
set of rotations that leads to each assignment.
\begin{center}
  \begin{tabular}{c|c}
  \multirow{2}{*}{Rotations}  &  Stable \\
                              &  Assignment  \\
  \hline\hline
  \multirow{2}{*}{$R_1, R_2, R_3, R_4$}     &  $(1,6), (2,9), (3,8)$ \\
                                            &  $(4,10), (5,7), (11,12)$ \\
  \hline
  \multirow{2}{*}{$R_1, R_2, R_3, R_4^d$}   &  $(1,10), (2,9), (3,8)$ \\
                                            &  $(4,6), (5,7), (11,12)$ \\
  \hline
  \multirow{2}{*}{$R_1, R_2, R_4, R_3^d$}   &  $(1,7), (2,6), (3,8)$ \\
                                            &  $(4,10), (5,9), (11,12)$ \\
  \hline
  \multirow{2}{*}{$R_1, R_3, R_4, R_2^d$}   &  $(1,6), (2,9), (3,8)$ \\
                                            &  $(4,12), (5,7), (10,11)$ \\
  \hline
  \multirow{2}{*}{$R_1, R_4, R_2^d, R_3^d$} &  $(1,7), (2,6), (3,8)$ \\
                                            &  $(4,12), (5,9), (10,11)$ \\
  \end{tabular}
\end{center}

The Hasse diagram of rotation poset of this roommate instance is
as follows.

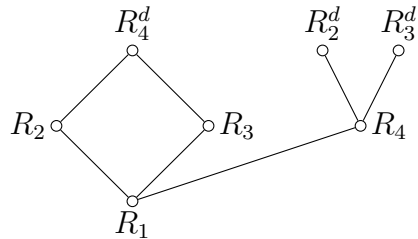
\begin{center}
\begin{tikzpicture}[inner sep=0.5mm]
   \node (R1) at (0,0) [circle,draw,label=below:$R_1$] {};

   \node (R2) at (-1,1) [circle, draw,label=left:$R_2$] {};
   \node (R3) at (1,1) [circle,draw,label=right:$R_3$] {};
   \node (R4) at (3,1) [circle,draw,label=right:$R_4$] {};

   \node (R4d) at (0,2) [circle,draw,label=above:$R_4^d$] {};
   \node (R2d) at (2.5,2) [circle,draw,label=above:$R_2^d$] {};
   \node (R3d) at (3.5,2) [circle,draw,label=above:$R_3^d$] {};

   \draw[-] (R1) to (R2);
   \draw[-] (R1) to (R3);
   \draw[-] (R1) to (R4);

   \draw[-] (R2) to (R4d);
   \draw[-] (R3) to (R4d);

   \draw[-] (R4) to (R2d);
   \draw[-] (R4) to (R3d);
\end{tikzpicture}
\end{center}

Finally, the graph $G(I)$ (recall the definition from Section~\ref{sect:ind-sets}) for
this instance is
as follows.

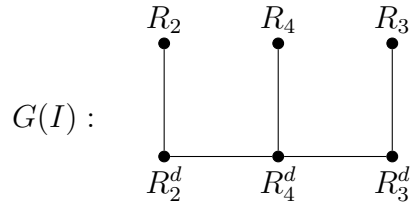
\begin{center}
\begin{tikzpicture} [inner sep=0.5mm]
   \node (R4d) at (0,0) [circle,draw,fill,label=below:$R_4^d$] {};
   \node (R2d) at (-1.5,0) [circle,draw,fill,label=below:$R_2^d$] {};
   \node (R3d) at (1.5,0) [circle, draw,fill,label=below:$R_3^d$] {};

   \node (R4) at (0,1.5) [circle,draw,fill,label=above:$R_4$] {};
   \node (R2) at (-1.5,1.5) [circle,draw,fill,label=above:$R_2$] {};
   \node (R3) at (1.5,1.5) [circle,draw,fill,label=above:$R_3$] {};

   \draw[-] (R2) to (R2d);
   \draw[-] (R3) to (R3d);
   \draw[-] (R4) to (R4d);
   \draw[-] (R4d) to (R2d);
   \draw[-] (R4d) to (R3d);

   \node[text width=1cm] at (-3,.5) {$G(I):$};
\end{tikzpicture}
\end{center}

Recall that the maximal independent sets in $G(I)$ are in $1$-$1$
correspondence with the stable roommate assignments.  These independent sets
can be read off directly from
the table above
using
the left-hand column, and deleting $R_1$ from the set of rotations, e.g.\
the third assignment $(1,7), (2,6), (3,8), (4,10), (5,9), (11,12)$ corresponds
to the maximal independent set $\{  R_2, R_4, R_3^d \}$ in $G(I)$.

\newpage
\section*{Appendix 2 --- Preference lists and rotations for the example given in Figure~\ref{fig:bip-graph1}.}

Here are the prefixes (from (\ref{eq:noisy-Q-pref-lists}) and (\ref{eq:noisy-P-pref-lists}))
for the example given in Figure~\ref{fig:bip-graph1}.
First, from the first $\rho$-cycle, $(1,2)$ and the first $\sigma$-cycle, $(3,4)$, we have the following lists.

$$\begin{array}{c|ll}
Q_{1} & P_{1} P_{2} \ \cdots \\
Q_{2} & P_{2} P_{1} \ \cdots \\
P_{1} & Q_{2} {P}_{3}\{P_{4}  \} Q_{1} \ \cdots \\
P_{2} & Q_{1} {P}_{4} Q_{2}\ \cdots \\
 \end{array}$$

 Then,  from the $\rho$-cycle $(3,4,5,6,7,8)$ and the corresponding $\sigma$-cycle
 $(1,2,9,10,13,14)$ we have the following.

 $$\begin{array}{c|ll}
Q_{3} & P_{3} P_{4} \ \cdots \\
Q_{4} & P_{4} P_{5} \ \cdots \\
Q_{5} & P_{5} P_{6} \ \cdots \\
Q_{6} & P_{6} P_{7} \ \cdots \\
Q_{7} & P_{7} P_{8} \ \cdots \\
Q_{8} & P_{8}  \{ P_7 P_6 P_5 P_4\} P_{3} \ \cdots \\
P_{3} & Q_{8} {P}_{1} \{  P_{14} Q_{7} P_{13} Q_{6} P_{10} Q_{5} P_{9} Q_{4} P_2  \} Q_{3}\ \cdots \\
P_{4} & Q_{3} {P}_{2} Q_{4}\ \cdots \\
P_{5} & Q_{4} {P}_{9} Q_{5}\ \cdots \\
P_{6} & Q_{5} {P}_{10} Q_{6}\ \cdots \\
P_{7} & Q_{6} {P}_{13} Q_{7}\ \cdots \\
P_{8} & Q_{7} {P}_{14} Q_{8}\ \cdots \\
\end{array}$$

\From the $\rho$-cycle $(9,10,11,12)$ and the corresponding $\sigma$-cycle $(5,6,15,16)$ we
have the following.

$$\begin{array}{c|ll}
Q_{9} & P_{9} P_{10} \ \cdots \\
Q_{10} & P_{10} P_{11} \ \cdots \\
Q_{11} & P_{11} P_{12} \ \cdots \\
Q_{12} & P_{12} \{P_{11} P_{10} \} P_{9} \ \cdots \\
P_{9} & Q_{12} {P}_{5} \{  P_{16} Q_{11} P_{15} Q_{10} P_{6}    \} Q_{9}\ \cdots \\
P_{10} & Q_{9} {P}_{6} Q_{10}\ \cdots \\
P_{11} & Q_{10} {P}_{15} Q_{11}\ \cdots \\
P_{12} & Q_{11} {P}_{16} Q_{12}\ \cdots \\
\end{array}$$

Similarly, from the $\rho$-cycle $(13,14,15,16)$ and the corresponding $\sigma$-cycle $(7,8,11,12)$
we have the following.

$$\begin{array}{c|ll}
Q_{13} & P_{13} P_{14} \ \cdots \\
Q_{14} & P_{14} P_{15} \ \cdots \\
Q_{15} & P_{15} P_{16} \ \cdots \\
Q_{16} & P_{16} \{P_{15} P_{14} \} P_{13} \ \cdots \\
P_{13} & Q_{16} {P}_{7} \{  P_{12} Q_{15} P_{11} Q_{14} P_{8}    \} Q_{13}\ \cdots \\
P_{14} & Q_{13} {P}_{8} Q_{14}\ \cdots \\
P_{15} & Q_{14} {P}_{11} Q_{15}\ \cdots \\
P_{16} & Q_{15} {P}_{12} Q_{16}\ \cdots \\
\end{array}$$

The short lists (from Equations (\ref{eq:unnoisy-Q-pref-lists}) and (\ref{eq:unnoisy-P-pref-lists}))
are therefore as follows.

$$\begin{array}{c|ll}
Q_{1} & P_{1} P_{2} \ \cdots \\
Q_{2} & P_{2} P_{1} \ \cdots \\
P_{1} & Q_{2} {P}_{3}  Q_{1} \ \cdots \\
P_{2} & Q_{1} {P}_{4} Q_{2}\ \cdots \\
 \end{array}$$

$$\begin{array}{c|ll}
Q_{3} & P_{3} P_{4} \ \cdots \\
Q_{4} & P_{4} P_{5} \ \cdots \\
Q_{5} & P_{5} P_{6} \ \cdots \\
Q_{6} & P_{6} P_{7} \ \cdots \\
Q_{7} & P_{7} P_{8} \ \cdots \\
Q_{8} & P_{8}  P_{3} \ \cdots \\
P_{3} & Q_{8} {P}_{1}  Q_{3}\ \cdots \\
P_{4} & Q_{3} {P}_{2}  Q_{4}\ \cdots \\
P_{5} & Q_{4} {P}_{9}  Q_{5}\ \cdots \\
P_{6} & Q_{5} {P}_{10} Q_{6}\ \cdots \\
P_{7} & Q_{6} {P}_{13} Q_{7}\ \cdots \\
P_{8} & Q_{7} {P}_{14} Q_{8}\ \cdots \\
\end{array}$$

$$\begin{array}{c|ll}
Q_{9} & P_{9} P_{10} \ \cdots \\
Q_{10} & P_{10} P_{11} \ \cdots \\
Q_{11} & P_{11} P_{12} \ \cdots \\
Q_{12} & P_{12}  P_{9} \ \cdots \\
P_{9} & Q_{12} {P}_{5} Q_{9}\ \cdots \\
P_{10} & Q_{9} {P}_{6} Q_{10}\ \cdots \\
P_{11} & Q_{10} {P}_{15} Q_{11}\ \cdots \\
P_{12} & Q_{11} {P}_{16} Q_{12}\ \cdots \\
\end{array}$$

$$\begin{array}{c|ll}
Q_{13} & P_{13} P_{14} \ \cdots \\
Q_{14} & P_{14} P_{15} \ \cdots \\
Q_{15} & P_{15} P_{16} \ \cdots \\
Q_{16} & P_{16}   P_{13} \ \cdots \\
P_{13} & Q_{16} {P}_{7}   Q_{13}\ \cdots \\
P_{14} & Q_{13} {P}_{8} Q_{14}\ \cdots \\
P_{15} & Q_{14} {P}_{11} Q_{15}\ \cdots \\
P_{16} & Q_{15} {P}_{12} Q_{16}\ \cdots \\
\end{array}$$

\end{document}